\documentclass[sigconf,  nonacm]{acmart}
\setkeys{acmart.cls}{balance=false}
\usepackage{subfigure}
\usepackage{array}
\usepackage{pdfpages}
\usepackage{tabularx}
\usepackage{textcomp}
\usepackage{stfloats}
\usepackage{url}
\usepackage{verbatim}
\usepackage{graphicx}
\usepackage{mathrsfs}
\usepackage[switch]{lineno}

\usepackage{xcolor} 
\usepackage{caption}
\usepackage{float} 
\usepackage{hyperref}
\usepackage{subcaption}
\usepackage{grffile}
\usepackage{algorithm,algorithmicx}
\usepackage{algpseudocode}
\usepackage{CJKutf8}
\usepackage{amsfonts,amsthm,amsmath}
\usepackage{multirow}
\usepackage{makecell}
\usepackage{pifont}
\usepackage{tikz}
\usepackage{mathrsfs}
\usepackage{threeparttable}
\usepackage{balance}

\PassOptionsToPackage{table,xcdraw}{xcolor}
\usepackage{colortbl}
\usepackage{xcolor}

\DeclareFontFamily{U}{rsfs}{\skewchar\font127 }
\DeclareFontShape{U}{rsfs}{m}{n}{%
   <-6> rsfs5
   <6-8> rsfs7
   <8-> rsfs10
}{}
\makeatletter
\newenvironment{breakablealgorithm}
{
		\begin{center}
			\refstepcounter{algorithm}
			\hrule height.8pt depth0pt \kern2pt
			\renewcommand{\caption}[2][\relax]{
				{\raggedright\textbf{\ALG@name~\thealgorithm} ##2\par}%
				\ifx\relax##1\relax 
				\addcontentsline{loa}{algorithm}{\protect\numberline{\thealgorithm}##2}%
				\else 
				\addcontentsline{loa}{algorithm}{\protect\numberline{\thealgorithm}##1}%
				\fi
				\kern2pt\hrule\kern2pt
			}
		}{
		\kern2pt\hrule\relax
	\end{center}
}
\makeatother
\DeclareFontFamily{U}{rsfs}{\skewchar\font127 }
\DeclareFontShape{U}{rsfs}{m}{n}{%
   <-6.5> rsfs5
   <6.5-8> rsfs7
   <8-> rsfs10
}{}
\newcommand\vldbdoi{XX.XX/XXX.XX}

\newcommand\vldbavailabilityurl{URL_TO_YOUR_ARTIFACTS}

\newtheorem{theorem}{Theorem}
\newtheorem{lemma}{Lemma}
\newtheorem{definition}{Definition}
\newtheorem{problem}{Problem}
\floatname{algorithm}{Algorithm} 

\DeclareFontFamily{U}{rsfs}{\skewchar\font127 }
\DeclareFontShape{U}{rsfs}{m}{n}{%
   <-6> rsfs5
   <6-8> rsfs7
   <8-> rsfs10
}{}

\begin{document}

\begin{sloppypar}
\title{EVA-S2PMLP: Secure and Scalable Two-Party MLP via Spatial Transformation}

\author{Shizhao Peng*}
\orcid{0000-0001-6333-5703}
\affiliation{%
  \institution{Beihang University}
  \streetaddress{37 Xueyuan Road}
  \city{Beijing}
  \country{China}
  \postcode{100191}
}
\email{by1806167@buaa.edu.cn}

\author{Shoumo Li}
\orcid{0009-0005-9228-2651}
\affiliation{%
  \institution{Beihang University}
  \streetaddress{37 Xueyuan Road}
  \city{Beijing}
  \country{China}
  \postcode{100191}
}
\email{22371327@buaa.edu.cn}

\author{Tianle Tao}
\orcid{0009-0005-1218-1923}
\affiliation{%
  \institution{Beihang University}
  \streetaddress{37 Xueyuan Road}
  \city{Beijing}
  \country{China}
  \postcode{100191}
}
\email{taotianle@buaa.edu.cn}

\begin{abstract}
Privacy-preserving neural network training in vertically partitioned scenarios is vital for secure collaborative modeling across institutions. This paper presents \textbf{EVA-S2PMLP}, an Efficient, Verifiable, and Accurate Secure Two-Party Multi-Layer Perceptron framework that introduces spatial-scale optimization for enhanced privacy and performance. To enable reliable computation under real-number domain, EVA-S2PMLP proposes a secure transformation pipeline that maps scalar inputs to vector and matrix spaces while preserving correctness. The framework includes a suite of atomic protocols for linear and non-linear secure computations, with modular support for secure activation, matrix-vector operations, and loss evaluation. Theoretical analysis confirms the reliability, security, and asymptotic complexity of each protocol. Extensive experiments show that EVA-S2PMLP achieves high inference accuracy and significantly reduced communication overhead, with up to $12.3\times$ improvement over baselines. Evaluation on benchmark datasets demonstrates that the framework maintains model utility while ensuring strict data confidentiality, making it a practical solution for privacy-preserving neural network training in finance, healthcare, and cross-organizational AI applications.
\end{abstract}

\maketitle


\section{Introduction}
In the past decade, a wide range of privacy-preserving machine learning frameworks have been proposed to enable secure training and inference of neural networks under secure multi-party computation (SMPC). These frameworks span various adversarial models, protocol abstractions, and system architectures.
From the perspective of computation paradigm and protocol primitives, SecureML~\cite{mohassel2017secureml} was an early attempt under the 2-party setting, integrating multiple SMPC techniques such as secret sharing for linear operations, garbled circuits for Boolean operations, and oblivious transfer or homomorphic encryption for preprocessing. However, limited protocol optimization led to high overhead and reduced precision.
Subsequent works aimed to improve performance and scalability. ABY3~\cite{mohassel2018aby3} introduced a three-party protocol supporting conversions among arithmetic, Boolean, and Yao circuits, and provided an efficient secret-shared multiplication protocol. SecureNN~\cite{wagh2019securenn} proposed optimized ReLU and maxpool implementations without garbled circuits, improving communication efficiency. FALCON~\cite{wagh2020falcon} and BLAZE~\cite{Patra_Suresh_2020} built on these ideas to further reduce communication cost and improve precision using optimized truncation protocols.
Other frameworks focused on usability and integration with existing ML ecosystems. TF-Encrypted~\cite{dahl2018private} supported TensorFlow-based secure training, while EzPC~\cite{Chandran_Gupta_Rastogi_Sharma_Tripathi_2019} provided a high-level language to abstract cryptographic complexity. CrypTen~\cite{knott2021crypten} and SecretFlow~\cite{ma2023secretflow} offered GPU acceleration and modular cryptographic backends.
Robustness under adversarial models also received attention. FLASH~\cite{Byali_Chaudhari_Patra_Suresh_2020}, Trident~\cite{Rachuri_Suresh_Chaudhari_2020}, and SWIFT~\cite{Koti_Pancholi_Patra_Suresh_2021} explored efficient secure inner product protocols under malicious or dishonest majority models. MP-SPDZ~\cite{keller2020mp} consolidated various 2PC/3PC protocols in both semi-honest and malicious settings.
Meanwhile, federated learning frameworks such as FATE~\cite{Liu_Fan_Chen_Xu_Yang_2021} and PySyft~\cite{Ziller_2021} addressed privacy at the system level by combining SMPC with differential privacy or homomorphic encryption. TenSEAL~\cite{benaissa2021tenseal} offered CKKS-based encrypted training, while Pencil~\cite{liu2024pencil} leveraged GPU-based homomorphic encryption in a vertical FL setting.
These frameworks collectively highlight the trend towards higher efficiency, better usability, and stronger adversarial resilience in privacy-preserving machine learning.
\section{Related Work}\label{Related Work}
\subsection{Data Disguising Techniques}
\noindent S2PM and S3PM are the foundational linear computation protocols within our framework, from which all other sub-protocols can be derived. Existing research on secure matrix multiplication primarily utilizes SMPC and data disguising techniques.
Frameworks such as Sharemind \cite{du2001privacy,bogdanov2008sharemind} achieve secure matrix multiplication by decomposing matrices into vector dot products $\alpha_i \cdot \beta_i$. Each entry $\alpha_i \cdot \beta_i$ is computed locally, and re-sharing in the ring $\mathbb{Z}_{2^{32}}$ is completed through six Du-Atallah protocols.
Other studies, including \cite{braun2022motion,wagh2019securenn,tan2021cryptgpu}, reduce the linear complexity of matrix multiplication by precomputing random triples $\langle a \rangle^s, \langle b \rangle^s, \langle c \rangle^s$ (where $S$ denotes additive secret sharing over $\mathbb{Z}_{2^l}$). 
In \cite{miller2021simple,furukawa2017high,guo2020efficient}, optimizations to the underlying ZeroShare protocol employ AES as a PRNG in ECB mode to perform tensor multiplication, represented as \textbf{$\langle z \rangle = \langle x \cdot y \rangle$}. 
The DeepSecure and XOR-GC frameworks \cite{10.1145/3195970.3196023,kolesnikov2008improved} utilize custom libraries and standard logic synthesis tools to parallelize matrix multiplication using GC in logic gate operations, enhancing computational efficiency. 
ABY3\cite{mohassel2018aby3} combines GC and SS methods, introducing a technique for rapid conversion between arithmetic, binary, and Yao's 3PC representations, achieving up to a $50\times$  reduction in communication for matrix multiplications.
Other approaches, including \cite{benaissa2021tenseal}, use CKKS encryption for vector-matrix multiplication, expanding ciphertext slots by replicating input vectors to accommodate matrix operations.
LibOTe \cite{libOTe} implements a highly efficient 1-out-of-n OT by adjusting the Diffie-Hellman key-exchange and optimizing matrix linear operations with a $64 \times 64 \rightarrow 128$-bit serial multiplier in $\mathbb{F}_{2^{255}-19}$ for precision.
The SPDZ framework and its upgrades \cite{keller2018overdrive,baum2019using} enhance efficiency and provable security by integrating hidden key generation protocols for BGV public keys, combining the strengths of HE, OT, and SS.
Frameworks like SecretFlow, Secure ML, Chameleon, and Delphi \cite{mishra2020delphi,ma2023secretflow} integrate sequential interactive GMW, fixed-point ASS, precomputed OT, and an optimized STP-based vector dot product protocol for matrix multiplication, achieving significant improvements in communication overhead and efficiency.\\

\subsection{Non-linear Operations in SMPC}
\noindent Nikolaenko \cite{nikolaenko2013privacy} proposes a server-based privacy-preserving linear regression protocol for horizontally partitioned data, combining linearly homomorphic encryption (LHE) and GC. 

For vertically partitioned data, Giacomelli and Gasc{\'o}n \cite{gascon2016privacy,giacomelli2018privacy} utilize Yao's circuit protocol and LHE, incurring high overhead and limited non-linear computation precision. 
In contrast, ABY3 \cite{mohassel2018aby3} reduces regression communication complexity using delayed re-sharing techniques. Building on this, Mohassel further improves linear regression accuracy with an approximate fixed-point multiplication method that avoids Boolean truncation \cite{mohassel2017secureml}.
Gilad \cite{gilad2019secure} presents the first three-server linear regression model, yet heavy GC use limits scalability due to high communication costs.
Liu \cite{liu2024pencil} combines DP with HE and SS to support linear regression across vertically and horizontally partitioned models, protecting model parameter privacy.
Rathee and Tan \cite{tan2021cryptgpu,rathee2020cryptflow2} leverage GPU acceleration and fixed-point arithmetic over shares, using 2-out-of-3 replicated secret sharing to support secure regression across three-party servers.
Ma \cite{ma2023secretflow} introduces the first SPU-based, MPC-enabled PPML framework, with compiler optimizations that substantially enhance training efficiency and usability in secure regression.

\section{System Framework and Objectives}\label{Preliminaries}
In this section, we present the design of the proposed EVA-S2PMLP framework for privacy-preserving multi-layer perceptron (MLP) training and inference under the two-party computation (2PC) setting.

\subsection{System Architecture of EVA-S2PMLP}
As shown in \autoref{fig:S2PMLP-Framework}, the EVA-S2PMLP framework consists of four secure low-level primitives (S2PM, S2PHP, S2PSCR, and S2PHHP), four secure high-level primitives (S2PHM, S2PRL, S2PSM, and S2PG-MLP), and two 2PC-based MLP protocols (S2PMLP-TR for training and S2PMLP-PR for prediction).

Upon receiving an MLP training or prediction request from the client, the computation task is distributed to two data owners. Each data owner locally encodes its private data and participates in the protocol by securely feeding the input into the corresponding MLP pipeline. These pipelines internally invoke the secure matrix multiplication and activation primitives to perform the forward and backward propagation steps. To ensure privacy, a data disguising technique is employed before entering the protocol, which randomizes the input data with pre-generated masks, preventing any raw data leakage during the interactive computation phase.

\begin{figure}[ht]
  \centering
  \includegraphics[width=1.0\hsize]{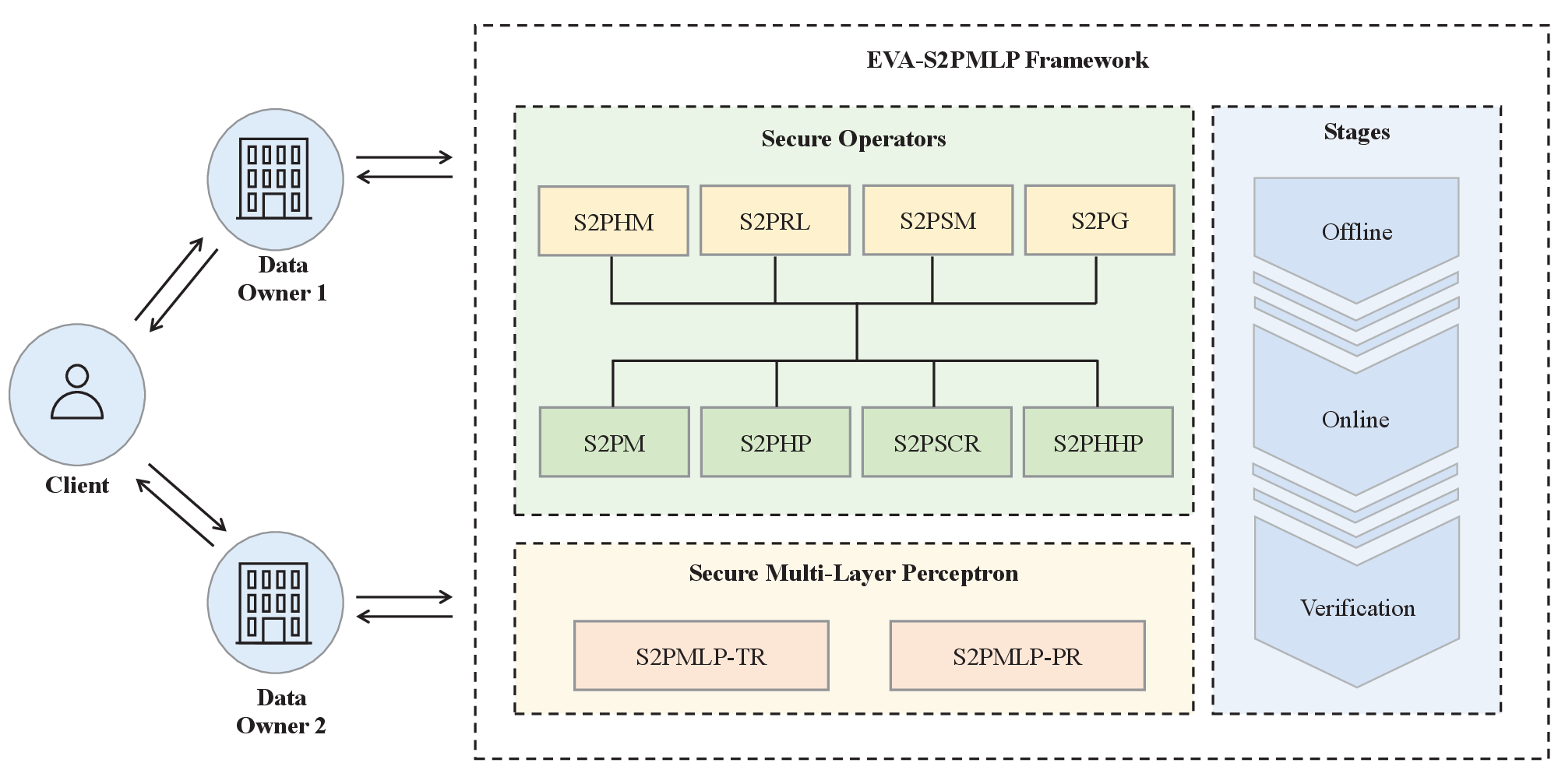}
  \Description{System architecture diagram of the EVA-S2PMLP framework.}
  \vspace{-0.6cm}
  \caption{System architecture of the EVA-S2PMLP framework}
  \label{fig:S2PMLP-Framework}
\end{figure}

Each protocol and operator in EVA-S2PMLP executes through three stages: (1) \textbf{Preprocessing}, where offline computations such as random matrix generation and format alignment are performed locally without communication; (2) \textbf{Online Computing}, which involves the core interactive execution between the two parties according to the predefined algorithm; and (3) \textbf{Verification}, where additional checks are performed using verification matrices to ensure the integrity of results. This design ensures correctness and verifiability while maintaining strong privacy guarantees throughout the lifecycle of model training or inference.

\subsection{Security Model}
\noindent For the definition of security, we follow a semi-honest model of S2PC using the criteria of computational indistinguishability between the view of ideal-world and the simulated views of real-world on a finite filed, and extended it to the scenario of a real number filed.

\begin{definition}[Semi-honest adversaries model \cite{evans2018pragmatic}]\label{def1}
    In a semi-honest adversary model, it is hypothesized that all participants follow the exact protocol during computation but may use input and intermediate results of their own to infer others’ original data.
\end{definition}

\begin{definition}[Computational Indistinguishability \cite{goldreich2004foundations}]\label{def2}
     A probability ensemble $X = \{X(a, n)\}_{a\in\{0,1\}^{*};n\in \mathbb{N}}$ is an infinite sequence of random variables indexed by $a\in \{0, 1\}^{*}$ and $n\in \mathbb{N}$. In the context of secure computation, the $a$ represents the parties’ inputs and $n$ denotes the security parameter. Two probability ensembles $X = \{X(a, n)\}_{a\in\{0,1\}^{*};n\in \mathbb{N}}$ and $Y = \{Y(a, n)\}_{a\in \{0,1\}^{*};n\in \mathbb{N}}$ are said to be computationally indistinguishable, denoted by $X\overset{c}{\equiv}Y$, if for every non-uniform polynomial-time algorithm D there exists a negligible function $\mu(\cdot)$ such that for every $a\in \{0, 1\}^{*}$ and every $n\in \mathbb{N}$,
\begin{align}
    |Pr[D(X(a,n))=1]-Pr[D(Y(a,n))=1]|\leq \mu(n)  
\end{align}
\end{definition}

\begin{definition}[Privacy in Semi-honest 2-Party Computation \cite{lindell2017simulate}]\label{def3}
    Let $f:\{0,1\}^{*}\times \{0,1\}^{*}\mapsto \{0,1\}^{*}\times\{0,1\}^{*}$ be a functionality, where $f_1(x,y)$ (resp.,$f_2(x,y)$) denotes the first(resp.,second) element of $f(x,y)$ and $\pi$ be a two-party protocol for computing $f$. The view of the first(resp., second) party during an execution of $\pi$ on $(x,y)$, denoted $VIEW^{\pi}_{1}(x, y)$ (resp., $VIEW^{\pi}_{2}(x, y)$), is $(x,r^1,m^1_1,\cdots,m^1_t)$ (resp.,$(y,r^2,m^2_1,\cdots ,m^2_t )$), where $r^1$ (resp.,$r^2$) represents the outcome of the first (resp.,second) party’s internal coin tosses, and $m^1_i$(resp.,$m^2_i$) represents the $i^{th}$ message it has received. The output of the first (resp.,second) party during an execution of $\pi$ on $(x,y)$, denoted $OUTPUT^{\pi}_1(x,y)$ (resp.,$OUTPUT^{\pi}_2(x,y)$), is implicit in the party’s view of the execution.
We say that $\pi$ privately computes $f(x,y)$ if there exist polynomial time algorithms, denoted $S_1$ and $S_2$ such that:
\begin{align}
 &\{(S_1(x,f_1(x, y)),f_2(x,y))\}_{x,y\in \{0,1\}^{*}}\nonumber \\
 &\overset{c}{\equiv} \{(VIEW^{\pi}_1(x, y), OUTPUT^{\pi}_2(x, y))\}_{x,y\in \{0,1\}^{*}} \nonumber
 \end{align}  
 \begin{align}
 &\{(f_1(x, y), S_2(y, f_2(x, y)))\}_{x,y\in \{0,1\}^{*}}\nonumber \\
 &\overset{c}{\equiv} \{(OUTPUT^{\pi}_1(x, y), VIEW^{\pi}_2(x, y))\}_{x,y\in \{0,1\}^{*}}
 \end{align}
 where $\overset{c}{\equiv}$ denotes computational indistinguishability and $|x|=|y|$. We stress that above $VIEW^{\pi}_{1}(x,y)$ and $VIEW^{\pi}_{2}(x,y)$, $OUTPUT^{\pi}_1(x,y)$ and $OUTPUT^{\pi}_2(x, y)$ are related random variables, defined as a function of the same random execution.
\end{definition}

\begin{definition}[Privacy in Semi-honest 3-party computation]\label{def4}
Let $f=(f_1,f_2,f_3)$ be a functionality. We say that $\pi$ privately computes $f(x,y,z)$ if there exist polynomial time algorithms, denoted $S_1$, $S_2$ and $S_3$ such that:
\begin{align}
 & \{(S_1(x,f_1(x, y, z)),f_2(x,y,z),f_3(x,y,z))\}_{x,y,z\in \{0,1\}^{*}}\nonumber \\
 & \overset{c}{\equiv} \{(VIEW^{\pi}_1(x, y, z), OUTPUT^{\pi}_2(x, y, z),\nonumber \\
 & OUTPUT^{\pi}_3(x, y, z)\}_{x,y,z\in \{0,1\}^{*}} \nonumber
 \end{align}  
 \begin{align}
 & \{(f_1(x, y, z),S_2(y,f_2(x,y,z)),f_3(x,y,z))\}_{x,y,z\in \{0,1\}^{*}}\nonumber \\
 & \overset{c}{\equiv} \{(OUTPUT^{\pi}_1(x, y, z), VIEW^{\pi}_2(x, y, z),\nonumber \\
 & OUTPUT^{\pi}_3(x, y, z)\}_{x,y,z\in \{0,1\}^{*}} \nonumber
 \end{align}  
 \begin{align}
 & \{(f_1(x, y, z),f_2(x,y,z),S_3(z,f_3(x,y,z)))\}_{x,y,z\in \{0,1\}^*{}}\nonumber \\
 & \overset{c}{\equiv} \{(OUTPUT^{\pi}_1(x, y, z), OUTPUT^{\pi}_2(x, y, z)\nonumber \\
 & VIEW^{\pi}_3(x, y, z)\}_{x,y,z\in \{0,1\}^{*}} 
 \end{align}
where, again, $\overset{c}{\equiv}$ denotes computational indistinguishability and $|x|=|y|=|z|$. $VIEW^{\pi}_{1}(x,y,z)$. $VIEW^{\pi}_{2}(x,y,z)$ and $VIEW^{\pi}_{3}(x,y,z)$, $OUTPUT^{\pi}_1(x,y,z)$, $OUTPUT^{\pi}_2(x, y, z)$ and $OUTPUT^{\pi}_3(x, y, z)$ are related random variables, defined as a function of the same random execution.
\end{definition}

\noindent This definition is for the general case of the real-ideal security paradigm defined in a formal language and for deterministic functions, as long as they can ensure that the messages $\{S_i(x,f_i(x)), (i\in 1,2...)\}$ generated by the simulator in the ideal-world are distinguishable form $\{view_i^\pi(x), (i\in 1,2...)\}$ in the real-world, then it can be shown that a protocol privately computes $f$ in a finite field. Furthermore, a heuristic model defined on a real number field is introduced as follows \cite{du2004privacy}:
\begin{definition}[Security Model in field of real number]\label{def5}
    All inputs in this model are in the real number field $\mathbb{R}$. Let $I_A$ and $I_B$ represent Alice's and Bob's private inputs, and $O_A$ and $O_B$ represent Alice's and Bob's outputs, respectively. Let $\pi$ denote the two-party computation involving Alice and Bob, where $(O_A, O_B)=\pi(I_A, I_B)$. Protocol $\pi$ is considered secure against dishonest Bob if there is an infinite number of $(I^{*}_A, O^{*}_A)$ pairs in $(\mathbb{R},\mathbb{R})$ such that $(O^{*}_A,O_B)=\pi(I^{*}_A,I_B)$. A protocol $\pi$ is considered secure against dishonest Alice if there is an infinite number of $(I^{*}_B, O^{*}_B)$ pairs in $(\mathbb{R},\mathbb{R})$ such that $(O_A, O^{*}_B)=\pi(I_A, I^{*}_B)$.
\end{definition}

\noindent A protocol is considered secure in the field of real numbers if, for any input/output combination $(I,O)$ from one party, there are an infinite number of alternative inputs in $\mathbb{R}$ from the second party that will result in $O$ from the first party's perspective given its own input $I$. From the adversary's point of view, this infinite number of the other party’s input/output represents a kind of stochastic indistinguishability in real number field, which is similar to computational indistinguishability in the real-ideal paradigm. Moreover, a simulator in the ideal world is limited to merely accessing the corrupted parties' input and output. In other words, the protocol $\pi$ is said to securely compute $f$ in the field of real numbers if, and only if, computational indistinguishability is achieved with any inputs from non-adversaries over a real number field, and the final outputs generated by the simulator are constant and independent from all inputs except for the adversaries.

\section{PROPOSED WORK}\label{Proposed-Work}
This section primarily introduces the detailed procedures of the basic protocols in the EVA-S2PMLP framework, including S2PRIP, S2PDRL, S2PRL, S2PHP, S2PHHP, S2PSCR, S2PSM, and S2PG-MLP. In addition, we provide a correctness analysis for these protocols.

\subsection{Secure Two-Party Row Inner Product Protocol (S2PRIP)}
The problem definition of S2PRIP is as follows:
\begin{problem}[Secure Two-Party Row Inner Product Protocol]\label{Problem-S2PRIP}
    Alice has a private matrix $A$, and Bob has a private matrix $B$, both of which have dimensions $n \times m$. They aim to perform a secure two-party row-wise inner product computation such that Alice obtains a single-column matrix $V_a$ and Bob obtains a single-column matrix $V_b$, satisfying $V_a + V_b = A \circledast B$.
\end{problem}

\textbf{Description of the S2PRIP Protocol:}
Similar to the S2PM protocol described in the appendix, the proposed S2PRIP protocol consists of three phases: the preprocessing phase (see Algorithm \ref{alg:S2PRIP-Preprocessing}), the online computation phase (see Algorithm \ref{alg:S2PRIP-Computing}), and the result verification phase (see Algorithm \ref{alg:S2PRIP-Verification}).

\textbf{Preprocessing Phase:} In Algorithm~\ref{alg:S2PRIP-Preprocessing}, the computation server (CS) generates a set of random private matrices $(R_a, r_a)$ and $(R_b, r_b)$ for Alice and Bob, respectively, to mask their input matrices $A$ and $B$. Additionally, for the purpose of subsequent result verification (see Algorithm~\ref{alg:S2PRIP-Verification}), the standard matrix $S_t = R_a \circledast R_b$ is also sent to Alice and Bob.

\begin{breakablealgorithm}
  \caption{S2PRIP CS Preprocessing Phase}
  \label{alg:S2PRIP-Preprocessing}
  \begin{algorithmic}[1] 
      \Require $n, m$
      \Ensure Alice $\Leftarrow$ $(R_a, r_a, S_t)$, Bob $\Leftarrow$ $(R_b, r_b, S_t)$
    
      \State  $R_a \gets$ generate random matrix \Comment{$R_a \in \mathbb{R}^{n \times m}$}
      \State  $R_b \gets$ generate random matrix \Comment{$R_b \in \mathbb{R}^{n \times m}$}
      \State  $S_t \gets R_a \circledast R_b$ \Comment{$S_t \in \mathbb{R}^{n \times 1}$}
      \State  $r_a, r_b \gets$ generate random matrices such that $r_a + r_b = S_t$ 
      \Statex \Comment{$r_a, r_b \in \mathbb{R}^{n \times 1}$}
      \State  Alice $\gets (R_a, r_a, S_t)$ \;
      \State  Bob $\gets (R_b, r_b, S_t)$ \;
      \State \Return $(R_a, r_a, S_t)$, $(R_b, r_b, S_t)$ \;
  \end{algorithmic}
\end{breakablealgorithm}

\textbf{Online Phase:} After the CS preprocessing phase, the online phase consists of a series of matrix computations, as described in Algorithm \ref{alg:S2PRIP-Computing}. Note that the final row-wise inner product $A \circledast B$ is masked by $V_a$ and $V_b$ to prevent Alice or Bob from knowing the actual result. The correctness of the result can be easily proven:  
$V_a + V_b = [(\hat{A} \circledast B + (r_a - V_b)) + r_a - (R_a \circledast \hat{B})] + V_b = [A \circledast B - V_b + (r_a + r_b - R_a \circledast R_b)] + V_b = A \circledast B$.

\begin{breakablealgorithm}
  \caption{S2PRIP Online Computation Phase}
  \label{alg:S2PRIP-Computing}
  \begin{algorithmic}[1] 
      \Require{$A \in \mathbb{R}^{n \times m}$ and $B \in \mathbb{R}^{n \times m}$}
      \Ensure{Alice $\Leftarrow (V_a, VF_a)$, Bob $\Leftarrow (V_b, VF_b)$}
    
      \State $\hat{A} = A + R_a$ and send $\hat{A} \Rightarrow$ Bob \Comment{$\hat{A} \in \mathbb{R}^{n \times m}$}
      \State $\hat{B} = B + R_b$ and send $\hat{B} \Rightarrow$ Alice \Comment{$\hat{B} \in \mathbb{R}^{n \times m}$}
      \State $V_b \gets$ generate random matrix \Comment{$V_b \in \mathbb{R}^{n \times 1}$}
      \State $VF_b = V_b - \hat{A} \circledast B$ \Comment{$VF_b \in \mathbb{R}^{n \times 1}$}
      \State $T = r_b - VF_b$ \Comment{$T \in \mathbb{R}^{n \times 1}$}
      \State Send $(VF_b, T) \Rightarrow$ Alice
      \State $V_a = T + r_a - (R_a \circledast \hat{B})$ \Comment{$V_a \in \mathbb{R}^{n \times 1}$}
      \State $VF_a = V_a + R_a \circledast \hat{B}$ and send $VF_a \Rightarrow$ Bob \Comment{$VF_a \in \mathbb{R}^{n \times 1}$}
      \State \Return $(V_a, VF_a)$, $(V_b, VF_b)$ 
    \end{algorithmic}
\end{breakablealgorithm}

\textbf{S2PRIP Verification Phase:} Similar to the S2PM result verification module in the appendix, we also propose a result verification algorithm for S2PRIP, as shown in Algorithm \ref{alg:S2PRIP-Verification}.

\begin{breakablealgorithm}
    \caption{S2PRIP Result Verification Phase} \label{alg:S2PRIP-Verification}
    \begin{algorithmic}[1]
        \Require{$VF_a, VF_b, S_t \in \mathbb{R}^{n \times 1}$} 
        \Ensure{Accept or Reject} 
        \For{$i=1:l$}
            \State Alice generates a vector $\hat{\delta_a} \in \mathbb{R}^{n \times 1}$, where each element is randomly chosen to be 0 or 1
            \State Alice then computes $E_r = (VF_a + VF_b - S_t) \odot \hat{\delta_a}$
            \If{$E_r\neq (0,0,\cdots,0)^T$}
            \State \Return Reject
            \EndIf
        \EndFor
        \State Bob repeats the same verification process as Alice
        \State \Return Accept
    \end{algorithmic}
\end{breakablealgorithm}

Similar to the verification analysis of S2PM, we can also derive that the verification failure probability of S2PRIP is $P_f(\text{S2PRIP}) \leq \frac{1}{4^l} \approx 9.09 \times 10^{-13} \ (l = 20)$. Based on the security analysis of S2PM, we can derive the following theorem:

\begin{theorem}\label{theorem:Secure-S2PRIP}
    The S2PRIP protocol is secure under the semi-honest adversarial model.
\end{theorem}

\subsection{Secure Two-Party DReLU Protocol (S2PDRL)}

The problem definition for S2PDRL is as follows:

\begin{problem}[Secure Two-Party ReLU Derivative]\label{Problem-S2PDRL}
    Alice has a private matrix $A$, and Bob has a private matrix $B$, both with dimensions $n \times m$. They want to perform a secure two-party ReLU derivative operation such that Alice obtains a matrix $V_a$ and Bob obtains a matrix $V_b$, satisfying $V_a = V_b = relu'(A+B)$.
\end{problem}

\textbf{S2PDRL Protocol Description:} In S2PDRL, Alice and Bob first convert their private matrices into vectors in row-major order, denoted as $\boldsymbol{a} = M2v(A)$ and $\boldsymbol{b} = M2v(B)$, respectively. Each of them generates a random positive real number, denoted as $p$ and $q$. Then, Alice splits each $a_i \in \boldsymbol{a}(1 \le i \le nm)$ randomly into $\rho$ real numbers $\alpha_{i}^{(1)}, \alpha_{i}^{(2)}, \cdots, \alpha_{i}^{(\rho)}$, forming a vector $\boldsymbol{\alpha_i} = p \cdot (\alpha_{i}^{(1)}, 1, \alpha_{i}^{(2)}, 1, \cdots, \alpha_{i}^{(\rho)}, 1)^T$. Alice then constructs a matrix $T_a$ by placing each vector $\boldsymbol{\alpha_i}$ as a row sequentially. Similarly, Bob splits each $b_i \in \boldsymbol{b}(1 \le i \le nm)$ randomly into $\rho$ real numbers $\beta_{i}^{(1)}, \beta_{i}^{(2)}, \cdots, \beta_{i}^{(\rho)}$, forming a vector $\boldsymbol{\beta_i} = q \cdot (1, \beta_{i}^{(1)}, 1, \beta_{i}^{(2)}, \cdots, 1, \beta_{i}^{(\rho)})^T$. Bob constructs a matrix $T_b$ by placing each vector $\boldsymbol{\beta_i}$ as a column sequentially. 
\begin{breakablealgorithm}
    \caption{S2PDRL}
    \label{alg:S2PDRL}
    \begin{algorithmic}[1]
        \Require {$A, B\in\mathbb{R}^{n \times m}$ and $\rho \ge 2$}
        \Ensure {$V_a + V_b = relu'(A + B)$ and $V_a, V_b \in \mathbb{R}^{n \times m}$}
        
        \State$\boldsymbol{a} = M2v(A)$ and $\boldsymbol{b} =M2v(B)$ \Comment{$\boldsymbol{a} \in \mathbb{R}^{nm\times 1}$, $\boldsymbol{b} \in \mathbb{R}^{nm\times 1}$}
        
        \State$p, q \gets$ Generate random positive real numbers \Comment{$p,q > 0$}
        \For{$i:=1$ \textbf{to} $nm$}
            \State$\boldsymbol{\alpha_i} = p \cdot (\alpha_{i}^{(1)}, 1, \alpha_{i}^{(2)}, 1, \cdots, \alpha_{i}^{(\rho)}, 1)^T$ 
            \Statex \Comment{$\boldsymbol{\alpha_i} \in \mathbb{R}^{2\rho \times 1}$, $\sum_{j=1}^{\rho}\alpha_{i}^{(j)} = a_i$}
            \State$\boldsymbol{\beta_i} = q \cdot (1, \beta_{i}^{(1)}, 1, \beta_{i}^{(2)}, \cdots, 1, \beta_{i}^{(\rho)})^T$ 
            \Statex \Comment{$\boldsymbol{\beta_i} \in \mathbb{R}^{2\rho \times 1}$, $\sum_{j=1}^{\rho}\beta_{i}^{(j)} = b_i$}
        \EndFor
        \State$T_a = [\boldsymbol{\alpha_1}, \boldsymbol{\alpha_2}, \cdots, \boldsymbol{\alpha_{nm}}]^T$ \Comment{$T_a \in \mathbb{R}^{nm \times 2\rho}$}
        \State$T_b = [\boldsymbol{\beta_1}, \boldsymbol{\beta_2}, \cdots, \boldsymbol{\beta_{nm}}]$ \Comment{$T_b \in \mathbb{R}^{2\rho \times nm}$}
        \State$U_a, U_b \gets \textbf{S2PM}(T_a, T_b)$ \Comment{$U_a, U_b \in \mathbb{R}^{nm \times nm}$}
        \State$\hat{V}_a = v2M(diag(U_a))\Rightarrow$ Bob \Comment{$\hat{V}_a \in \mathbb{R}^{n \times m}$}
        \State$\hat{V}_b = v2M(diag(U_b))\Rightarrow$ Alice \Comment{$\hat{V}_b \in \mathbb{R}^{n \times m}$}
        \State$V_a = V_b = relu'(\hat{V}_a + \hat{V}_b)$ \Comment{$V_a, V_b \in \mathbb{R}^{n \times m}$}
        \State\Return $V_a$, $V_b$
    \end{algorithmic}
\end{breakablealgorithm}
Alice and Bob jointly compute $T_a \times T_b = U_a + U_b$ using the S2PM protocol. They then extract the diagonal elements of $U_a$ and $U_b$, reconstruct matrices $\hat{V}_a$ and $\hat{V}_b$ with the same dimensions as their original private matrices $A$ and $B$, denoted as $\hat{V}_a = v2M(diag(U_a))$ and $\hat{V}_b = v2M(diag(U_b))$. Alice sends $\hat{V}_a$ to Bob, and Bob sends $\hat{V}_b$ to Alice. Finally, Alice computes $V_a = relu'(\hat{V}_a + \hat{V}_b)$, and Bob computes $V_b = relu'(\hat{V}_a + \hat{V}_b)$. The detailed process is shown in Algorithm \ref{alg:S2PDRL}. It is easy to verify that $V_a + V_b = relu'(v2M(diag(U_a + U_b))) = relu'(v2M(diag(T_a \times T_b))) = relu'(v2M(({\boldsymbol{\alpha_1}}^T \cdot \boldsymbol{\beta_1}, {\boldsymbol{\alpha_2}}^T \cdot \boldsymbol{\beta_2}, \cdots, {\boldsymbol{\alpha_{nm}}}^T \cdot \boldsymbol{\beta_{nm}})^T)) = relu'(v2M(pq \cdot (\sum_{i=1}^{\rho} \alpha_{i}^{(1)} + \sum_{i=1}^{\rho} \beta_{i}^{(1)}, \sum_{i=1}^{\rho} \alpha_{i}^{(2)} + \sum_{i=1}^{\rho} \beta_{i}^{(2)}, \cdots, \sum_{i=1}^{\rho} \alpha_{i}^{(nm)} + \sum_{i=1}^{\rho} \beta_{i}^{(nm)})^T)) = relu'(v2M((a_1 + b_1, a_2 + b_2, \cdots, a_{nm} + b_{nm})^T)) = relu'(v2M(\boldsymbol{a} + \boldsymbol{b})) = relu'(A + B)$.

\textbf{Optimization of the S2PDRL Protocol:} S2PDRL can be optimized using the S2PRIP protocol for parallelization, as shown in Algorithm \ref{alg:S2PDRL-opt}. The main idea of the algorithm remains unchanged, but redundant computations are eliminated by utilizing the secure row-wise inner product operation. The undetected anomaly situation in S2PDRL only occurs when the sub-protocol S2PRIP fails in anomaly detection. Therefore, the probability that S2PDRL fails to detect a computational anomaly is $P_f(\text{S2PDRL}) = P_f(\text{S2PRIP}) \leq \frac{1}{4^l} \approx 9.09 \times 10^{-13} \ (l=20)$, which is sufficiently small to be negligible.

We can prove the security of the S2PDRL protocol using a security proof process similar to that of S2PHP.

\begin{theorem}\label{theorem:Secure-S2PDRL}
    The S2PDRL protocol is secure under the semi-honest adversarial model.
\end{theorem}

\begin{proof}
Since the result of S2PDRL is a public value and does not require protection, we only need to analyze whether the private inputs $A$ and $B$ are exposed during the execution of the protocol. Considering that the online computation phase of the S2PDRL protocol is the same as that of S2PHP, with the only difference being in the offline local processing, we can similarly derive, based on Theorem~\ref{theorem:Secure-S2PHP}, that the online computation of S2PDRL is secure under the semi-honest adversarial model. Thus, we only need to consider whether the public disclosure of the result affects security. 

Since Alice and Bob each locally generate a random number, denoted as $p$ and $q$, respectively, the public computation result only reveals the value of $pq \cdot (A + B)$, which does not allow either party to deduce the private input of the other. Therefore, under the semi-honest adversarial model, the S2PDRL protocol $f(A,B) = \text{relu}'(A + B)$ is secure.
\end{proof}

\begin{breakablealgorithm}
    \caption{Optimized S2PDRL}
    \label{alg:S2PDRL-opt}
    \begin{algorithmic}[1]
        \Require {$A, B \in \mathbb{R}^{n \times m}$ and $\rho \ge 2$}
        \Ensure {$V_a + V_b = relu'(A + B)$ and $V_a, V_b \in \mathbb{R}^{n \times m}$}
        
        \State $\boldsymbol{a} = M2v(A)$ and $\boldsymbol{b} = M2v(B)$ \Comment{$\boldsymbol{a} \in \mathbb{R}^{nm \times 1}$, $\boldsymbol{b} \in \mathbb{R}^{nm \times 1}$}
        
        \State $p, q \gets$ Generate random positive real numbers \Comment{$p, q > 0$}
        \For{$i := 1$ \textbf{to} $nm$}
            \State $\boldsymbol{\alpha_i} = p \cdot (\alpha_{i}^{(1)}, 1, \alpha_{i}^{(2)}, 1, \cdots, \alpha_{i}^{(\rho)}, 1)^T$ 
            \Statex \Comment{$\boldsymbol{\alpha_i} \in \mathbb{R}^{2\rho \times 1}$, $\sum_{j=1}^{\rho}\alpha_{i}^{(j)} = a_i$}
            \State $\boldsymbol{\beta_i} = q \cdot (1, \beta_{i}^{(1)}, 1, \beta_{i}^{(2)}, \cdots, 1, \beta_{i}^{(\rho)})^T$ 
            \Statex \Comment{$\boldsymbol{\beta_i} \in \mathbb{R}^{2\rho \times 1}$, $\sum_{j=1}^{\rho}\beta_{i}^{(j)} = b_i$}
        \EndFor
        \State $T_a = [\boldsymbol{\alpha_1}, \boldsymbol{\alpha_2}, \cdots, \boldsymbol{\alpha_{nm}}]^T$ \Comment{$T_a \in \mathbb{R}^{nm \times 2\rho}$}
        \State $T_b = [\boldsymbol{\beta_1}, \boldsymbol{\beta_2}, \cdots, \boldsymbol{\beta_{nm}}]^T$ \Comment{$T_b \in \mathbb{R}^{nm \times 2\rho}$}
        \State $U_a, U_b \gets \textbf{S2PRIP}(T_a, T_b)$ \Comment{$U_a, U_b \in \mathbb{R}^{nm \times 1}$}
        \State $\hat{V}_a = reshape(U_a, (n, m)) \Rightarrow$ Bob \Comment{$\hat{V}_a \in \mathbb{R}^{n \times m}$}
        \State $\hat{V}_b = reshape(U_b, (n, m)) \Rightarrow$ Alice \Comment{$\hat{V}_b \in \mathbb{R}^{n \times m}$}
        \State $V_a = V_b = relu'(\hat{V}_a + \hat{V}_b)$ \Comment{$V_a, V_b \in \mathbb{R}^{n \times m}$}
        \State \Return $V_a$, $V_b$
    \end{algorithmic}
\end{breakablealgorithm}

\subsection{Secure Two-Party ReLU Protocol (S2PRL)}

The problem definition of S2PRL is as follows:

\begin{problem}[Secure Two-Party ReLU Function]\label{Problem-S2PRL}
    Alice has a private matrix $A$, and Bob has a private matrix $B$, both of which are $n \times m$ in dimension. They want to securely compute the ReLU function on the sum of their matrices such that Alice receives a matrix $V_a$ and Bob receives a matrix $V_b$, satisfying $V_a + V_b = \text{relu}(A + B)$.
\end{problem}

\begin{breakablealgorithm}
    \caption{S2PRL}
    \label{alg:S2PRL}
    \begin{algorithmic}[1]
        \Require Matrices $A, B \in \mathbb{R}^{n \times m}$ and $\rho \ge 2$
        \Ensure $V_a + V_b = \text{relu}(A + B)$, where $V_a, V_b \in \mathbb{R}^{n \times m}$
        
        \State $U_a, U_b \gets \textbf{S2PDRL}(A, B, \rho)$ \Comment{$U_a, U_b \in \mathbb{R}^{n \times m}$}
        \State $V_a = U_a \odot A$, $V_b = U_b \odot B$ \Comment{$V_a, V_b \in \mathbb{R}^{n \times m}$}
        \State \Return $V_a, V_b$
    \end{algorithmic}
\end{breakablealgorithm}
\textbf{S2PRL Protocol Description:}
In S2PRL, Alice and Bob jointly compute $U_a = U_b = \text{relu}'(A + B)$ using S2PDRL, and then they individually compute $V_a = U_a \odot A$ and $V_b = U_b \odot B$. The detailed procedure is shown in Algorithm \ref{alg:S2PRL}. It is easy to verify that $V_a + V_b = U_a \odot A + U_b \odot B = \text{relu}'(A + B) \odot (A + B) = \text{relu}(A + B)$. The undetected anomaly situation in S2PRL only occurs when the sub-protocol S2PDRL fails in anomaly detection. Therefore, the probability that S2PRL fails to detect a computational anomaly is $P_f(\text{S2PRL}) = P_f(\text{S2PDRL}) \leq \frac{1}{4^l} \approx 9.09 \times 10^{-13} \ (l=20)$. This failure probability is sufficiently small to be negligible.

The security of the S2PRL protocol can be proven based on the security of S2PDRL, as stated in the following theorem.

\begin{theorem}\label{theorem:Secure-S2PRL}
    The S2PRL protocol is secure under the semi-honest adversarial model.
\end{theorem}

\subsection{Secure Two-Party SoftMax Protocol (S2PSM)}

We first present the design of the S2PHP protocol. Subsequently, we describe the design of the S2PHHP protocol and the S2PSCR protocol, and finally, we provide the design of the S2PSM protocol.

\subsubsection{Secure Two-Party Matrix Hadamard Product Protocol (S2PHP)}

The problem definition of S2PHP (Matrix) is as follows:

\begin{problem}[Secure Two-Party Matrix Hadamard Product]\label{Problem-S2PHP-Matrix}
    Alice has a private matrix $A$, and Bob has a private matrix $B$, both of which are $n \times m$ in dimension. They wish to securely compute the Hadamard product of the matrices such that Alice obtains a matrix $V_a$ and Bob obtains a matrix $V_b$, satisfying $V_a + V_b = A \odot B$.
\end{problem}

\textbf{S2PHP Protocol Description:}  
In S2PHP, Alice and Bob first convert their private matrices into vectors using row-major order, denoted as $\boldsymbol{a} = M2v(A)$ and $\boldsymbol{b} = M2v(B)$. Then, Alice splits each $a_i \in \boldsymbol{a}$ ($1 \leq i \leq nm$) into $\rho$ random real numbers, forming a vector $\boldsymbol{\alpha_i}$. She replicates $\boldsymbol{\alpha_i}$ $\rho$ times to form a vector $\boldsymbol{\alpha_i^*}$. Subsequently, Alice uses each vector $\boldsymbol{\alpha_i^*}$ to construct a matrix $T_a$, with each vector as a row.
\begin{breakablealgorithm}
    \caption{S2PHP}
    \label{alg:S2PHP-Matrix}
    \begin{algorithmic}[1]
        \Require $A, B \in \mathbb{R}^{n \times m}$ and $\rho \geq 2$
        \Ensure $V_a + V_b = A \odot B$ and $V_a, V_b \in \mathbb{R}^{n \times m}$
        
        \State $\boldsymbol{a} = M2v(A)$ and $\boldsymbol{b} = M2v(B)$ \Comment{$\boldsymbol{a} \in \mathbb{R}^{nm \times 1}$, $\boldsymbol{b} \in \mathbb{R}^{nm \times 1}$}
        \For{$i := 1$ \textbf{to} $nm$}
            \State $\boldsymbol{\alpha_i} = (\alpha_{i}^{(1)}, \alpha_{i}^{(2)}, \cdots, \alpha_{i}^{(\rho)})^T$ \Comment{$\boldsymbol{\alpha_i} \in \mathbb{R}^{\rho \times 1}$, $\sum_{j=1}^{\rho}{\alpha_{i}^{(j)}} = a_i$}
            \State $\boldsymbol{\alpha_i^*} = [\boldsymbol{\alpha_i}^T, \boldsymbol{\alpha_i}^T, \cdots, \boldsymbol{\alpha_i}^T]^T$ \Comment{$\boldsymbol{\alpha_i^*} \in \mathbb{R}^{\rho^2 \times 1}$}
        \EndFor
        \State $T_a = [\boldsymbol{\alpha_1^*}, \boldsymbol{\alpha_2^*}, \cdots, \boldsymbol{\alpha_{nm}^*}]^T$ \Comment{$T_a \in \mathbb{R}^{nm \times \rho^2}$}
        
        \For{$i := 1$ \textbf{to} $nm$}
            \State $\boldsymbol{\beta_i} = (\beta_{i}^{(1)}, \beta_{i}^{(2)}, \cdots, \beta_{i}^{(\rho)})^T$ \Comment{$\boldsymbol{\beta_i} \in \mathbb{R}^{\rho \times 1}$, $\sum_{j=1}^{\rho}{\beta_{i}^{(j)}} = b_i$}
            \State $perms(\boldsymbol{\beta_i}) = \{\boldsymbol{\beta_i^{(1)}}, \boldsymbol{\beta_i^{(2)}}, \cdots, \boldsymbol{\beta_i^{(\rho!)}}\}$
            \State $T_i \gets$ randomly select $\rho$ vectors from $perms(\boldsymbol{\beta_i})$ and concatenate them \Comment{$T_i \in \mathbb{R}^{\rho \times \rho}$}
            \State $\boldsymbol{\beta_i^*} = M2v(T_i)$ \Comment{$\boldsymbol{\beta_i^*} \in \mathbb{R}^{\rho^2 \times 1}$}
        \EndFor
        
        \State $T_b = [\boldsymbol{\beta_1^*}, \boldsymbol{\beta_2^*}, \cdots, \boldsymbol{\beta_{nm}^*}]$ \Comment{$T_b \in \mathbb{R}^{\rho^2 \times nm}$}
        \State $U_a, U_b \gets$ \textbf{S2PM}$(T_a, T_b)$ \Comment{$U_a, U_b \in \mathbb{R}^{nm \times nm}$}
        \State $V_a = v2M(diag(U_a))$ \Comment{$V_a \in \mathbb{R}^{n \times m}$}
        \State $V_b = v2M(diag(U_b))$ \Comment{$V_b \in \mathbb{R}^{n \times m}$}
        \State \Return $V_a, V_b$
    \end{algorithmic}
\end{breakablealgorithm}

Bob splits each $b_i \in \boldsymbol{b}$ ($1 \leq i \leq nm$) into $\rho$ random real numbers, forming a vector $\boldsymbol{\beta_i}$. He generates all permutations of $\boldsymbol{\beta_i}$, denoted as $perms(\boldsymbol{\beta_i})$, and randomly selects $s$ vectors from these permutations to concatenate into a matrix $T_i$. Bob then converts $T_i$ into a vector $\boldsymbol{\beta_i^*}$ using row-major order, denoted as $\boldsymbol{\beta_i^*} = M2v(T_i)$. Using each vector $\boldsymbol{\beta_i^*}$, Bob constructs a matrix $T_b$, with each vector as a column.

Alice and Bob jointly compute $T_a \times T_b = U_a + U_b$ using the S2PM protocol. They then extract the main diagonal elements of $U_a$ and $U_b$, reconstructing matrices $V_a$ and $V_b$ with the same dimensions as their original private matrices $A$ and $B$. This is done by $V_a = v2M(diag(U_a))$ and $V_b = v2M(diag(U_b))$. The detailed procedure is shown in Algorithm \ref{alg:S2PHP-Matrix}. It is easy to verify that $V_a + V_b = v2M(\text{diag}(U_a + U_b)) = v2M(\text{diag}(T_a \times T_b)) = v2M(({\boldsymbol{\alpha_1^*}}^T \cdot \boldsymbol{\beta_1^*}, {\boldsymbol{\alpha_2^*}}^T \cdot \boldsymbol{\beta_2^*}, \cdots, {\boldsymbol{\alpha_{nm}^*}}^T \cdot \boldsymbol{\beta_{nm}^*})^T) = v2M((\sum_{i=1}^{\rho}{\alpha_{i}^{(1)}} \cdot \sum_{i=1}^{\rho}{\beta_{i}^{(1)}}, \sum_{i=1}^{\rho}{\alpha_{i}^{(2)}} \cdot \sum_{i=1}^{\rho}{\beta_{i}^{(2)}}, \cdots, \sum_{i=1}^{\rho}{\alpha_{i}^{(nm)}} \cdot \sum_{i=1}^{\rho}{\beta_{i}^{(nm)}})^T) = v2M((a_1 \cdot b_1, a_2 \cdot b_2, \cdots, a_{nm} \cdot b_{nm})^T) = v2M(\boldsymbol{a} \odot \boldsymbol{b}) = A \odot B$. 

\textbf{Optimization of the S2PHP Protocol}: Similar to the optimization method of the S2PDRL protocol, the S2PHP protocol can also be optimized in parallel using the S2PRIP protocol to eliminate redundant computations, as illustrated in Algorithm~\ref{alg:S2PHP-Matrix-opt}. Similar to the verification analysis of S2PDRL, we can derive that the verification failure probability of S2PHP is $P_f(\text{S2PHP}) = P_f(\text{S2PRIP}) \leq \frac{1}{4^l} \approx 9.09 \times 10^{-13} \ (l=20)$.

\begin{breakablealgorithm}
    \caption{Optimized S2PHP}
    \label{alg:S2PHP-Matrix-opt}
    \begin{algorithmic}[1]
        \Require $A, B \in \mathbb{R}^{n \times m}$ and $\rho \geq 2$
        \Ensure $V_a + V_b = A \odot B$ and $V_a, V_b \in \mathbb{R}^{n \times m}$
        
        \State $\boldsymbol{a} = M2v(A)$ and $\boldsymbol{b} = M2v(B)$ \Comment{$\boldsymbol{a} \in \mathbb{R}^{nm \times 1}$, $\boldsymbol{b} \in \mathbb{R}^{nm \times 1}$}
        \For{$i := 1$ \textbf{to} $nm$}
            \State $\boldsymbol{\alpha_i} = (\alpha_{i}^{(1)}, \alpha_{i}^{(2)}, \cdots, \alpha_{i}^{(\rho)})^T$ \Comment{$\boldsymbol{\alpha_i} \in \mathbb{R}^{\rho \times 1}$, $\sum_{j=1}^\rho \alpha_{i}^{(j)} = a_i$}
            \State $\boldsymbol{\alpha_i^*} = [\boldsymbol{\alpha_i}^T, \boldsymbol{\alpha_i}^T, \cdots, \boldsymbol{\alpha_i}^T]^T$ \Comment{$\boldsymbol{\alpha_i^*} \in \mathbb{R}^{\rho^2 \times 1}$}
        \EndFor
        \State $T_a = [\boldsymbol{\alpha_1^*}, \boldsymbol{\alpha_2^*}, \cdots, \boldsymbol{\alpha_{nm}^*}]^T$ \Comment{$T_a \in \mathbb{R}^{nm \times \rho^2}$}
        
        \For{$i := 1$ \textbf{to} $nm$}
            \State $\boldsymbol{\beta_i} = (\beta_{i}^{(1)}, \beta_{i}^{(2)}, \cdots, \beta_{i}^{(\rho)})^T$ \Comment{$\boldsymbol{\beta_i} \in \mathbb{R}^{\rho \times 1}$, $\sum_{j=1}^\rho \beta_{i}^{(j)} = b_i$}
            \State $perms(\boldsymbol{\beta_i}) = \{\boldsymbol{\beta_i^{(1)}}, \boldsymbol{\beta_i^{(2)}}, \cdots, \boldsymbol{\beta_i^{(\rho!)}}\}$
            \State $T_i \gets$ randomly select $\rho$ vectors from $perms(\boldsymbol{\beta_i})$ and concatenate them \Comment{$T_i \in \mathbb{R}^{\rho \times \rho}$}
            \State $\boldsymbol{\beta_i^*} = M2v(T_i)$ \Comment{$\boldsymbol{\beta_i^*} \in \mathbb{R}^{\rho^2 \times 1}$}
        \EndFor
        
        \State $T_b = [\boldsymbol{\beta_1^*}, \boldsymbol{\beta_2^*}, \cdots, \boldsymbol{\beta_{nm}^*}]^T$ \Comment{$T_b \in \mathbb{R}^{nm \times \rho^2}$}
        \State $U_a, U_b \gets$ \textbf{S2PRIP}$(T_a, T_b)$ \Comment{$U_a, U_b \in \mathbb{R}^{nm \times 1}$}
        \State $V_a = reshape(U_a, (n, m))$ \Comment{$V_a \in \mathbb{R}^{n \times m}$}
        \State $V_b = reshape(U_b, (n, m))$ \Comment{$V_b \in \mathbb{R}^{n \times m}$}
        \State \Return $V_a, V_b$
    \end{algorithmic}
\end{breakablealgorithm}

\textbf{Security Analysis of the S2PHP Protocol:}
According to the definition of semi-honest adversarial security in the two-party computation model, let $f = (f_1, f_2)$ be a polynomial-time probabilistic function, and let $\pi$ be a secure two-party protocol for computing the function $f$. We consider the protocol $\pi$ to securely compute the function $f$ if we can construct two simulators, $S_1$ and $S_2$, in the ideal world such that the following relationships hold simultaneously:
\begin{align}
    \{S_1(x, f_1(x, y))\}_{x, y \in \{0, 1\}^*} \overset{c}{\equiv} \{view_1^{\pi}(x, y)\}_{x, y \in \{0, 1\}^*} \nonumber \\
    \{S_2(y, f_2(x, y))\}_{x, y \in \{0, 1\}^*} \overset{c}{\equiv} \{view_2^{\pi}(x, y)\}_{x, y \in \{0, 1\}^*} \label{12}
\end{align}
and for each participant, the output remains consistent regardless of changes to the input of the other participant:
\begin{align}
    f_1(x, y) \equiv f_1(x, y^*) \nonumber \\
    f_2(x, y) \equiv f_2(x^*, y) \label{13}
\end{align}
In this way, we transform the security proof into a constructive problem. In the subsequent proof, we will use the following lemma~\ref{lemma1}:

\begin{lemma}\label{lemma1}
For a linear system $A \cdot X = B$, if $rank(A) = rank(A|B) < n$ (where $n$ is the number of rows in matrix $X$, and $A|B$ is the augmented matrix), then the linear system has infinitely many solutions \cite{Shores_2018}.
\end{lemma}

Next, we will provide the security proof of S2PHP under the semi-honest adversarial model. In this proof, the process of transforming matrices $A$ and $B$ into matrices $T_a$ and $T_b$ in Algorithm~\ref{alg:S2PHP-Matrix} is denoted as $T_a = RA2T(A, \rho)$ and $T_b = RB2T(B, \rho)$, respectively. Additionally, we use $R = M2diag(A)$ to represent the generation of a matrix $R$, where the main diagonal elements are composed of all elements of matrix $A$ in row-major order, and all other elements are randomly generated.

\begin{theorem}\label{theorem:Secure-S2PHP}
    The S2PHP protocol is secure under the semi-honest adversarial model.
\end{theorem}

\begin{proof}
    We construct two simulators $S_1$ and $S_2$  to prove the above theorem.\\
    
    \textbf{Adversary Alice}: Assume Alice is the adversary, which means we need to construct a simulator $S_1$ to simulate ${view}_1^\pi (x,y) = (A,\rho,R_a, r_a, S_t, r; \hat{B}, T, VF_b)$ such that $S_1(x, f_1(x,y))$ is indistinguishable from ${view}_1^\pi (x,y)$. Formally, $S_1$ receives ($A$, $\rho$, $R_a$, $r_a$, $S_t$, $V_a$) and a random tape $r$, and then proceeds as follows:
    
    \begin{enumerate}
        \item $S_1$ uses the random tape $r$ to generate the matrix $T_a = RA2T(A, \rho)$ and computes $\hat{A} = T_a + R_a$ and $r_b = S_t - r_a$.
        \item $S_1$ generates two random matrices $B'$ and $V_b'$ such that $A \odot B' = V_a + V_b'$, then generates $T_b' = RB2T(B', \rho)$ and $U_b' = M2diag(V_b')$. \label{proof1_step2}
        \item $S_1$ generates a random matrix $R_b'$ satisfying $R_a \times R_b'=S_t$, and computes $\hat{B}' = T_b' + R_b'$, $T' = \hat{A} \times T_b' + r_b - U_b'$, and $VF_b' = U_b' - \hat{A} \times T_b'$. \label{proof1_step3}
        \item $S_1$ outputs $S_1(x, f_1(x,y)) = (A,\rho,R_a, r_a, S_t, r; \hat{B}', T', VF_b')$.
    \end{enumerate}
    In step \ref{proof1_step2}, we observe that there are infinitely many pairs $(B', V_b')$ satisfying the equation $A \odot B' = V_a + V_b'$, so $B'$ and $V_b'$ are simulatable, denoted as $\{B', V_b'\} \overset{c}{\equiv} \{B, V_b\}$. Consequently, $T_b'$ and $U_b'$ are also simulatable since $T_b' = RB2T(B', \rho)$ and $U_b' = M2diag(V_b')$. In step \ref{proof1_step3}, by Lemma \ref{lemma1}, we know that when $R_a$ is a rank-deficient matrix, there exist infinitely many $R_b'$ satisfying $R_a \times R_b' = S_t$, so $R_b'$ is simulatable. Considering that the variables $\hat{B}', T', VF_b'$ are computed from the simulatable variables $T_b', R_b', U_b'$, it is easy to verify that these variables are also simulatable. Therefore, we have $\{(A,\rho,R_a, r_a, S_t, r; \hat{B}', T', VF_b')\} \overset{c}{\equiv} \{(A,\rho,R_a, r_a, S_t, r; \hat{B}, T, VF_b)\}$, i.e., $\{S_1(x, f_1(x,y))\}_{x,y} \overset{c}{\equiv} \{view_1^{\pi}(x,y)\}_{x,y}$. Furthermore, $U_a' = T' + r_a - R_a \times \hat{B}' = \hat{A} \times {T_b}' - R_a \times \hat{B}' + r_b + r_a - {U_b}' = T_a \times T_b' - {U_b}'$. Therefore, $V_a' = diag2M(U_a') = diag2M(T_a \times T_b' - {U_b}')$. Based on the correctness of S2PHP, we conclude that $V_a' = V_a$, which means Alice's output $V_a$ remains unchanged for any input from Bob, i.e., $f_1(x,y) \equiv f_1(x,y^*)$.\\

    \textbf{Adversary Bob}: Assume Bob is the adversary, which means we need to construct a simulator $S_2$ to simulate ${view}_2^\pi (x,y)=(B,\rho,R_b, r_b, S_t, r_0, r_1; \hat{A}, VF_a)$ such that $S_2(y, f_2(x,y))$ is indistinguishable from ${view}_2^\pi (x,y)$. Formally, $S_2$ receives $(B,\rho,R_b, r_b, S_t, V_b)$ and two random tapes $r_0, r_1$, and then proceeds as follows:
    
    \begin{enumerate}
        \item $S_2$ uses the random tape $r_0$ to generate the matrix $T_b = RB2T(B, \rho)$ and computes $\hat{B} = T_b + R_b$. Meanwhile, $S_2$ uses the random tape $r_1$ to generate the matrix $U_b$ satisfying $V_b = diag2M(U_b)$.
        \item $S_2$ generates two random matrices $A'$ and $V_a'$ satisfying $A' \odot B = V_a' + V_b$, then generates $T_a' = RA2T(A', \rho)$ and $U_a' = M2diag(V_a')$. \label{proof2_step2}
        \item $S_2$ generates a random matrix $R_a'$ satisfying $R_a' \times R_b = S_t$, and computes $\hat{A}' = T_a' + R_a'$ and $VF_a' = U_a' + R_a' \times \hat{B}$. \label{proof2_step3}
        \item $S_2$ outputs $S_2(y, f_2(x,y)) = (B,\rho,R_b, r_b, S_t, r_0, r_1; \hat{A}', VF_a')$.
    \end{enumerate}
    In step \ref{proof2_step2}, we observe that there are infinitely many pairs $(A', V_a')$ satisfying the equation $A' \odot B = V_a' + V_b$, so $A'$ and $V_a'$ are simulatable, denoted as $\{A', V_a'\} \overset{c}{\equiv} \{A, V_a\}$. Consequently, $T_a'$ and $U_a'$ are also simulatable since $T_a' = RA2T(A', \rho)$ and $U_a' = M2diag(V_a')$. In step \ref{proof2_step3}, by Lemma \ref{lemma1}, we know that when $R_b$ is a rank-deficient matrix, there exist infinitely many $R_a'$ satisfying $R_a' \times R_b = S_t$, so $R_a'$ is simulatable. Considering that the variables $\hat{A}', VF_a'$ are computed from the simulatable variables $T_a', R_a', U_a'$, it is easy to verify that these variables are also simulatable. Therefore, we have $\{(B,\rho,R_b, r_b, S_t, r_0, r_1; \hat{A}', VF_a')\} \overset{c}{\equiv} \{(B,\rho,R_b, r_b, S_t, r_0, r_1; \hat{A}, VF_a)\}$, i.e., $\{S_2(y, f_2(x,y))\}_{x,y} \overset{c}{\equiv} \{view_2^{\pi}(x,y)\}_{x,y}$. Furthermore, since $V_b' = diag2M(U_b) = V_b$, Bob's output $V_b$ remains unchanged for any input from Alice, i.e., $f_2(x,y) \equiv f_2(x^*,y)$.
\end{proof}

\subsubsection{Secure Two-Party Hybrid Hadamard Product Problem (S2PHHP)}

The problem definition of S2PHHP is as follows:

\begin{problem}[Secure Two-Party Hybrid Hadamard Product Problem]\label{Problem-S2PHHP}
    Alice has a pair of private matrices $(A_1, A_2)$, and Bob has a pair of private matrices $(B_1, B_2)$, where $A_1, B_1, A_2, B_2 \in \mathbb{R}^{n\times m}$. They wish to perform a mixed Hadamard product computation, i.e., $f[(A_1, A_2), (B_1, B_2)] = (A_1 + B_1) \odot (A_2 + B_2)$, where Alice obtains matrix $V_a$, Bob obtains matrix $V_b$, and they satisfy $V_a + V_b = (A_1 + B_1) \odot (A_2 + B_2)$.
\end{problem}

\textbf{S2PHHP Protocol Description:} In S2PHHP, Alice and Bob first independently compute $V_{a0} = A_1 \odot A_2$ and $V_{b0} = B_1 \odot B_2$. Then, using the S2PHP (matrix) protocol, they jointly compute $V_{a1} + V_{b1} = A_1 \odot B_2$ and also compute $V_{b2}$ and $V_{a2}$ (both in $\mathbb{R}^{n \times m}$). Finally, Alice computes $V_a = V_{a0} + V_{a1} + V_{a2}$, and Bob computes $V_b = V_{b0} + V_{b1} + V_{b2}$. It is easy to verify that $V_a + V_b = V_{a0} + V_{a1} + V_{a2} + V_{b0} + V_{b1} + V_{b2} = A_1 \odot A_2 + B_1 \odot B_2 + A_1 \odot B_2 + B_1 \odot A_2 = (A_1 + B_1) \odot (A_2 + B_2)$. S2PHHP will only fail to detect an anomaly if all the sub-protocols S2PHP fail to detect the anomaly. Therefore, the probability of S2PHHP failing to detect a computational anomaly is $P_f(S2PHHP) = P_f(S2PHP)^2 \leq \left(\frac{1}{4^l}\right)^2 \approx 8.27\times 10^{-25}$ (where $l = 20$).

\begin{breakablealgorithm}
    \caption{S2PHHP}
    \label{alg:S2PHHP}
    \begin{algorithmic}[1]
        \Require $A_1, B_1, A_2, B_2 \in \mathbb{R}^{n \times m}$, $\rho \geq 2$
        \Ensure $V_a, V_b \in \mathbb{R}^{n \times m}$
        
        \State $V_{a0} = A_1 \odot A_2$ \Comment{$V_{a0} \in \mathbb{R}^{n \times m}$}
        \State $V_{b0} = B_1 \odot B_2$ \Comment{$V_{b0} \in \mathbb{R}^{n \times m}$}
        \State $V_{a1}, V_{b1} \gets$ \textbf{S2PHP}$(A_1, B_2, \rho)$ \Comment{$V_{a1}, V_{b1} \in \mathbb{R}^{n \times m}$}
        \State $V_{b2}, V_{a2} \gets$ \textbf{S2PHP}$(B_1, A_2, \rho)$ \Comment{$V_{a2}, V_{b2} \in \mathbb{R}^{n \times m}$}
        
        \State $V_a = V_{a0} + V_{a1} + V_{a2}$ \Comment{$V_a \in \mathbb{R}^{n \times m}$}
        \State $V_b = V_{b0} + V_{b1} + V_{b2}$ \Comment{$V_b \in \mathbb{R}^{n \times m}$}
        \State \Return $V_a, V_b$
    \end{algorithmic}
\end{breakablealgorithm}

We prove the security of the S2PHHP protocol based on the Universal Composability (UC) framework\cite{10.5555/874063.875553}. Therefore, we use the following lemma \cite{bogdanov2008sharemind}.

\begin{lemma}\label{lemma2}
    If all sub-protocols are perfectly simulatable, then the protocol itself is perfectly simulatable.
\end{lemma}

\begin{theorem}\label{theorem:Secure-S2PHHP}
    The S2PHHP protocol is secure under the semi-honest adversarial model.
\end{theorem}

\begin{proof}
    The S2PHHP protocol is implemented in a hybrid model where two parallel calls to the S2PHP protocol are made to compute $V_a + V_b = V_{a0} + V_{a1} + V_{a2} + V_{b0} + V_{b1} + V_{b2} = A_1 \odot A_2 + B_1 \odot B_2 + S2PHP_1(A_1, B_2, \rho) + S2PHP_2(B_1, A_2, \rho)$.
    According to Lemma~\ref{lemma2}, the security of S2PHHP can be reduced to the compositional security of the two parallel calls to S2PHP. Since the outputs $(V_{a1}, V_{a2})$ and $(V_{b1}, V_{b2})$ from the two S2PHP calls are respectively held by Alice and Bob, and according to Theorem~\ref{theorem:Secure-S2PHP}, these outputs in the real-world view are computationally indistinguishable from the simulated ideal-world view. Considering that $A_1 \odot A_2$ and $B_1 \odot B_2$ are locally computed, it is straightforward to prove that the summation results $V_a = V_{a0} + V_{a1} + V_{a2}$ and $V_b = V_{b0} + V_{b1} + V_{b2}$ are simulatable and indistinguishable from the real-world view. Therefore, in the semi-honest adversarial model, the S2PHHP protocol $f[(A_1, A_2), (B_1, B_2)] = (A_1 + B_1) \odot (A_2 + B_2)$ is secure.
\end{proof}

\subsubsection{Secure Two-Party Matrix Reciprocal Protocol based on Scale Collapse (S2PSCR)}

The problem definition of S2PSCR is as follows:

\begin{problem}[Secure Two-Party Matrix Reciprocal Protocol based on Scale Collapse]\label{Problem-S2PSCR}
    Alice has a private matrix $A$, and Bob has a private matrix $B$, both with dimensions $n \times m$. They want to perform a secure two-party matrix reciprocal operation such that Alice obtains a matrix $V_a$ and Bob obtains a matrix $V_b$, satisfying $V_a + V_b = \frac{1}{A + B}$.
\end{problem}

\textbf{S2PSCR Protocol Description:}
In S2PSCR, Alice first randomly generates a matrix $P$ with no zero elements and computes the matrix $I_a = P \odot A$. Similarly, Bob randomly generates a matrix $Q$ with no zero elements and computes the matrix $I_b = Q \odot B$. Then, Alice and Bob jointly compute $I_a \odot Q = U_{a1} + U_{b1} = (P \odot A) \odot Q$ and $P \odot I_b = U_{a2} + U_{b2} = P \odot (B \odot Q)$ using the S2PHP (matrix) protocol. Alice computes $U_a = U_{a1} + U_{a2}$ and sends $U_a$ to Bob. Upon receiving $U_a$, Bob computes $T = U_a + U_{b1} + U_{b2}$ and $I_b^* = \frac{Q}{T}$. Finally, Alice and Bob jointly compute $P \odot I_b^* = V_a + V_b = \frac{P \odot Q}{T}$ using S2PHP. The detailed process is shown in Algorithm \ref{alg:S2PSCR}. It is easy to verify that $V_a + V_b = P \odot I_b^* = \frac{P\odot Q}{T}=\frac{P\odot Q}{(U_{a1} + U_{b1}) + (U_{a2} + U_{b2})}=\frac{P\odot Q}{P\odot Q \odot(A+B)}=\frac{1}{A+B}$. S2PSCR will only fail to detect an anomaly if all the sub-protocols S2PHP fail to detect the anomaly. Therefore, the probability of S2PSCR failing to detect a computational anomaly is $P_f(S2PSCR) = P_f(S2PHP)^3 \leq \left(\frac{1}{4^l}\right)^3 \approx 7.52\times 10^{-37} \ (l = 20)$.

\begin{breakablealgorithm}
    \caption{S2PSCR}
    \label{alg:S2PSCR}
    \begin{algorithmic}[1]
        \Require $A, B \in \mathbb{R}^{n \times m}$, $\rho \geq 2$
        \Ensure $V_a + V_b = \frac{1}{A + B}$, $V_a, V_b \in \mathbb{R}^{n \times m}$
        
        \State $P, Q \gets$ randomly generate matrices without zero elements \Comment{$P, Q \in \mathbb{R}^{n \times m}$}
        \State $I_a = P \odot A$, $I_b = Q \odot B$ \Comment{$I_a, I_b \in \mathbb{R}^{n \times m}$}
        \State $U_{a1}, U_{b1} \gets$ \textbf{S2PHP}$(I_a, Q, \rho)$ \Comment{$U_{a1}, U_{b1} \in \mathbb{R}^{n \times m}$}
        \State $U_{a2}, U_{b2} \gets$ \textbf{S2PHP}$(P, I_b, \rho)$ \Comment{$U_{a2}, U_{b2} \in \mathbb{R}^{n \times m}$}
        \State $U_a = U_{a1} + U_{a2} \Rightarrow$ BoB \Comment{$U_a \in \mathbb{R}^{n \times m}$}
        \State $T = U_a + U_{b1} + U_{b2}$, $I_b^* = \frac{Q}{T}$ \Comment{$T, I_b^* \in \mathbb{R}^{n \times m}$}
        \State $V_a, V_b \gets$ \textbf{S2PHP}$(P, I_b^*, \rho)$ \Comment{$V_a, V_b \in \mathbb{R}^{n \times m}$}
        \State \Return $V_a, V_b$
    \end{algorithmic}
\end{breakablealgorithm}

We base our security proof of S2PSCR on the Universal Composability (UC) security framework.

\begin{theorem}\label{theorem:Secure-S2PSCR}
    The S2PSCR protocol, denoted as $f(A,B)= \frac{1}{A+B}$, is secure under the semi-honest adversarial model.
\end{theorem}

\begin{proof}
    The S2PSCR protocol is implemented in a hybrid model, where the protocol first calls two parallel instances of the S2PHP protocol to compute $T = P\odot(A+B)\odot Q = (P\odot A)\odot Q + P \odot (B\odot Q) = S2PHP_1(P\odot A,Q,s) + S2PHP_2(P,B\odot Q,s)$. Finally, a third S2PHP is sequentially called to compute $\frac{1}{A+B} = S2PHP_3(P, \frac{Q}{T}, s)$.
    According to Lemma \ref{lemma2}, the security of S2PSCR can be reduced to the sequential compositional security of these two steps. Given that the outputs of the two S2PHP calls in the first step, $(U_{a1}, U_{a2})$ and $(U_{b1}, U_{b2})$, are respectively held by Alice and Bob, and based on Theorem \ref{theorem:Secure-S2PHP}, these outputs in the real-world view are computationally indistinguishable from the simulated ideal-world view. Therefore, it is straightforward to prove that the summation result $T= (U_{a1}+U_{a2})+(U_{b1}+U_{b2})$ is simulatable. The S2PHP in the second step can also be perfectly simulated in the ideal world, making it indistinguishable from the real-world view. Thus, in the semi-honest adversarial model, the S2PSCR protocol $f(A,B)=\frac{1}{A+B}$ is secure.
\end{proof}

\subsubsection{Secure Two-Party Softmax Protocol (S2PSM)}

The problem definition of S2PSM is as follows:

\begin{problem}[Secure Two-Party Softmax Function]\label{Problem-S2PSM}
    Alice has a private matrix $A$, and Bob has a private matrix $B$, both with dimensions $n \times m$. They want to perform a secure two-party matrix Softmax function such that Alice obtains a matrix $V_a$ and Bob obtains a matrix $V_b$, satisfying $V_a + V_b = \text{softmax}(A+B)$.
\end{problem}

\textbf{S2PSM Protocol Description:}
In S2PSM, Alice and Bob first compute $I_a = e^A$ and $I_b = e^B$ locally. Then, they jointly compute $I_a \odot I_b = U_{a1} + U_{b1} = e^{A+B}$ using the S2PHP protocol. Next, Alice and Bob each compute the horizontal sum of matrices $U_{a1}$ and $U_{b1}$ locally, denoted as $U_{a2} = \text{hsum}(U_{a1})$ and $U_{b2} = \text{hsum}(U_{b1})$. They then use the S2PSCR protocol to jointly compute $\frac{1}{U_{a2}+U_{b2}}=U_{a3}+U_{b3}$, and each horizontally replicates $U_{a3}$ and $U_{b3}$ $m$ times to match the dimensions of $A$ and $B$, denoted as $U_{a4} = \text{hcopy}(U_{a3})$ and $U_{b4} = \text{hcopy}(U_{b3})$. Finally, Alice and Bob use the S2PHHP protocol to jointly compute $(U_{a1} + U_{b1}) \odot (U_{a4} + U_{b4}) = V_a + V_b$. The detailed process is shown in Algorithm \ref{alg:S2PSM}. It is easy to verify that $V_a + V_b = (U_{a1} + U_{b1}) \odot (U_{a4} + U_{b4}) = I_a \odot I_b \odot \text{hcopy}(U_{a3}+U_{b3})=\frac{e^{A+B}}{\text{hcopy}(\text{hsum}(U_{a1}+U_{b1}))}=\frac{e^{A+B}}{\text{hcopy}(\text{hsum}(e^{A+B}))}=\text{softmax}(A+B)$. S2PSM will only fail to detect an anomaly if all its sub-protocols fail to detect the anomaly. Therefore, the probability of S2PSM failing to detect a computational anomaly is $P_f(S2PSM) = P_f(S2PHP) \cdot P_f(S2PSCR) \cdot P_f(S2PHHP) \leq \left(\frac{1}{4^l}\right)^6 \approx 5.66 \times 10^{-73} \ (l = 20)$.

\begin{breakablealgorithm}
    \caption{S2PSM}
    \label{alg:S2PSM}
    \begin{algorithmic}[1]
        \Require $A, B \in \mathbb{R}^{n \times m}$, $\rho \geq 2$
        \Ensure $V_a + V_b = softmax(A + B)$, $V_a, V_b \in \mathbb{R}^{n \times m}$
        
        \State $I_a = e^A$, $I_b = e^B$ \Comment{$I_a, I_b \in \mathbb{R}^{n \times m}$}
        \State $U_{a1}, U_{b1} \gets$ \textbf{S2PHP}$(I_a, I_b, \rho)$ \Comment{$U_{a1}, U_{b1} \in \mathbb{R}^{n \times m}$}
        \State $U_{a2} = hsum(U_{a1})$, $U_{b2} = hsum(U_{b1})$ \Comment{$U_{a2}, U_{b2} \in \mathbb{R}^{n \times 1}$}
        \State $U_{a3}, U_{b3} \gets$ \textbf{S2PSCR}$(U_{a2}, U_{b2}, \rho)$ \Comment{$U_{a3}, U_{b3} \in \mathbb{R}^{n \times 1}$}
        \State $U_{a4} = hcopy(U_{a3})$, $U_{b4} = hcopy(U_{b3})$ \Comment{$U_{a4}, U_{b4} \in \mathbb{R}^{n \times m}$}
        \State $V_a, V_b \gets$ \textbf{S2PHHP}$((U_{a1}, U_{a4}), (U_{b1}, U_{b4}), \rho)$ \Comment{$V_a, V_b \in \mathbb{R}^{n \times m}$}
        \State \Return $V_a, V_b$
    \end{algorithmic}
\end{breakablealgorithm}

We can prove the security of S2PSM based on the UC framework.

\begin{theorem}\label{theorem:Secure-S2PSM}
    The S2PSM protocol is secure under the semi-honest adversarial model.
\end{theorem}

\subsection{Secure Two-Party Gradient Protocol for Multi-Layer Perceptrons (S2PG-MLP)}

The problem definition of S2PG-MLP is as follows:

\begin{problem}[Secure Two-Party Multi-Layer Perceptron Gradient Computation Process]\label{Problem-S2PG-MLP}
    Let the network parameters be given by $D=\{d_0, d_1, \cdots, d_L\}$, which represent the number of neurons in each layer, including the input, hidden, and output layers, where $L(L\ge2)$ denotes the number of layers excluding the input layer. For the $l$-th layer of the model $(1\le l < L)$, assume that Alice and Bob each hold parts of the model parameters of layer $l + 1$, $\hat{W}_a^{(l + 1)}, \hat{W}_b^{(l + 1)}$, and the gradients $G_a^{(l+1)}, G_b^{(l+1)}$, as well as the intermediate input variables of layer $l$, $X_a^{(l)}, X_b^{(l)}$. They aim to perform privacy-preserving multi-layer perceptron gradient computation, so that Alice and Bob ultimately obtain $G_a^{(l)}$ and $G_b^{(l)}$, satisfying $G_a^{(l)} + G_b^{(l)} = G^{(l)}$, which represents the model gradient of layer $l$.
\end{problem}

\textbf{S2PG-MLP Protocol Description:}
In the gradient computation of the $l$-th layer $(1\le l < L)$ during MLP backpropagation, the inputs are the model parameters of layer $l+1$ held by both parties, $\hat{W}_a^{(l + 1)}, \hat{W}_b^{(l + 1)}\in\mathbb{R}^{n\times d_L}$, and the gradients $G_a^{(l+1)}, G_b^{(l+1)}\in\mathbb{R}^{n\times d_{l+1}}$, the intermediate variables of layer $l$, $X_a^{(l)}, X_b^{(l)}\in\mathbb{R}^{n\times d_l}$, and a splitting parameter $\rho \ge 2$. The outputs are the gradients of layer $l$, $G_a^{(l)}, G_b^{(l)}\in\mathbb{R}^{n\times d_l}$. Alice and Bob first process $\hat{W}_a^{(l + 1)}, \hat{W}_b^{(l + 1)}$ by removing the first row and transposing them, denoted as $W_a = \text{dtrans}(\hat{W}_a^{(l + 1)})$ and $W_b = \text{dtrans}(\hat{W}_b^{(l + 1)})$. Then, Alice and Bob jointly compute $\text{relu}'(X_a^{(l)}+X_b^{(l)}) = F_a = F_b$ using the S2PDRL protocol, and compute $(G_a^{(l+1)} + G_b^{(l+1)})\times (W_a + W_b) = U_a + U_b$ using the S2PHM protocol. Finally, both parties locally compute the gradients of layer $l$, $G_a^{(l)} = F_a \odot U_a$ and $G_b^{(l)} = F_b \odot U_b$. The detailed process is shown in Algorithm \ref{alg:S2PG-MLP}. We analyze the probability of anomaly detection failure in S2PG-MLP. S2PG-MLP will only fail to detect an anomaly if all the sub-protocols S2PDRL and S2PHM fail to detect the anomaly. Therefore, the probability of S2PG-MLP failing to detect a computational anomaly is $P_f(S2PG\text{-}MLP) = P_f(S2PDRL) \cdot P_f(S2PHM) \leq \left(\frac{1}{4^l}\right)^3 \approx 7.52 \times 10^{-37} \ (l = 20)$.

\begin{breakablealgorithm}
    \caption{S2PG-MLP}
    \label{alg:S2PG-MLP}
    \begin{algorithmic}[1]
        \Require $\hat{W}_a^{(l + 1)}, \hat{W}_b^{(l + 1)}, G_a^{(l+1)}, G_b^{(l+1)}, X_a^{(l)}, X_b^{(l)}, \rho$
        \Ensure $G_a^{(l)}, G_b^{(l)}$
        
        \State $W_a = dtrans(\hat{W}_a^{(l + 1)})$, $W_b = dtrans(\hat{W}_b^{(l + 1)})$ 
        \Statex \Comment{$W_a, W_b \in \mathbb{R}^{d_{l + 1} \times d_l}$}
        \State $F_a, F_b \gets$ \textbf{S2PDRL}$(X_a^{(l)}, X_b^{(l)}, \rho)$ \Comment{$F_a, F_b \in \mathbb{R}^{n \times d_l}$}
        \State $U_a, U_b \gets$ \textbf{S2PHM}$((G_a^{(l+1)}, W_a), (G_b^{(l+1)}, W_b))$ 
        \Statex \Comment{$U_a, U_b \in \mathbb{R}^{n \times d_l}$}
        \State $G_a^{(l)} = F_a \odot U_a$, $G_b^{(l)} = F_b \odot U_b$ \Comment{$G_a^{(l)}, G_b^{(l)} \in \mathbb{R}^{n \times d_l}$}
        \State \Return $G_a^{(l)}, G_b^{(l)}$
    \end{algorithmic}
\end{breakablealgorithm}

\textbf{Correctness}: For simplicity, we assume that the entire training set is used in each training round (i.e., $B = n$). Using $D=\{d_0, d_1, \cdots, d_L\}$ to represent the network structure, in each round of backpropagation, for the $L$-th layer, we have:$G_a^{(L)} + G_b^{(L)}=Y_a^{(L)} +  Y_b^{(L)} - Y=Y^{(L)} - Y = G^{(L)}$, and for the $l$-th layer $(1 \le l < L)$, we have: $G_a^{(l)} + G_b^{(l)} = F_a \odot U_a + F_b \odot U_b = ((G_a^{(l+1)} + G_b^{(l+1)}) \times (W_a + W_b)) \odot \text{relu}'(X_a^{(l)} + X_b^{(l)}) = (G^{(l+1)} \times \text{dtrans}(\hat{W}^{(l+1)})) \odot \text{relu}'(X^{(l)}) = G^{(l)}$. This follows the centralized gradient propagation formula, thus proving the correctness of the gradient protocol.

We can prove the security of S2PG-MLP based on the UC framework.

\begin{theorem}\label{theorem:Secure-S2PG-MLP}
    The S2PG-MLP protocol is secure under the semi-honest adversarial model.
\end{theorem}

\section{Secure Two-Party MLP Collaborative Modeling Based on Spatial Scale Optimization}
This section introduces the proposed secure two-party multi-layer perceptron (MLP) training and prediction algorithms (S2PMLP-TR and S2PMLP-PR), along with their corresponding correctness analyses.

\subsection{Secure Two-Party MLP Model Training}
The problem definition of S2PMLP-TR is as follows:

\begin{problem}[Secure Two-Party Multi-Layer Perceptron Training Process]\label{Problem-S2PMLP-TR}
    For a scenario involving heterogeneous distributed data, Alice and Bob each hold part of the training dataset, denoted as $X_a$ and $X_b$, respectively, and both parties share the labels $Y$. Using a batch size of $B$, a learning rate of $\eta$, and a total number of iterations $t$, the network parameters are represented as $D=\{d_0, d_1, \cdots, d_L\}$, where each $d_i$ represents the number of neurons in the $i$-th layer, including the input, hidden, and output layers, and $L$ is the number of layers excluding the input layer. The goal is to perform privacy-preserving MLP training based on the above parameters, so that Alice and Bob ultimately obtain $\hat{W}_a=\{\hat{W}_a^{(1)},\hat{W}_a^{(2)},\cdots, \hat{W}_a^{(L)}\}$ and $\hat{W}_b=\{\hat{W}_b^{(1)},\hat{W}_b^{(2)},\cdots, \hat{W}_b^{(L)}\}$, satisfying: $\hat{W} = \{\hat{W}^{(1)},\hat{W}^{(2)},\cdots, \hat{W}^{(L)}\} = \{\hat{W}_a^{(1)} + \hat{W}_b^{(1)}, \hat{W}_a^{(2)} + \hat{W}_b^{(2)},\cdots, \hat{W}_a^{(L)} + \hat{W}_b^{(L)}\}$, which represents the trained model parameters.
\end{problem}

\textbf{S2PMLP-TR Protocol Description:}
In MLP model training, the overall training process for each batch can be divided into two parts: forward propagation and backward propagation. Similarly, S2PMLP-TR can also be split into two parts: S2PMLP-TR forward propagation and S2PMLP-TR backward propagation. Embedding these two parts into the overall training process forms the complete S2PMLP-TR protocol.

\begin{breakablealgorithm}
    \caption{S2PMLP-TR}
    \label{alg:S2PMLP-TR}
    \begin{algorithmic}[1]
        \Require $L$, $D$, $X_a, X_b, Y, B, t, \eta$, $\rho$
        \Ensure $\hat{W}_a, \hat{W}_b$

        \State $N = \lceil \frac{n}{B} \rceil$
        \State $n_1, n_2, \cdots, n_{N-1} = B$ and $n_N = n - (N-1) \cdot B$
        \State \{$Y^{(1)}, Y^{(2)}, \cdots, Y^{(N)}$\} $\gets$ Split $Y$ sequentially into $N$ parts 
        \Statex \Comment{$Y^{(i)} \in \{0, 1\}^{n_i \times d_L}$}
        \State $\hat{X}_a, \hat{X}_b = addcol(X_a, X_b)$  \Comment{$\hat{X}_a, \hat{X}_b \in \mathbb{R}^{n \times (d_0 + 1)}$}
        \State \{$\hat{X}_a^{(1)}, \hat{X}_a^{(2)}, \cdots, \hat{X}_a^{(N)}$\} $\gets$ Split $\hat{X}_a$ into $N$ parts 
        \Statex \Comment{$\hat{X}_a^{(i)} \in \mathbb{R}^{n_i \times (d_0 + 1)}$}
        \State \{$\hat{X}_b^{(1)}, \hat{X}_b^{(2)}, \cdots, \hat{X}_b^{(N)}$\} $\gets$ Split $\hat{X}_b$ into $N$ parts 
        \Statex \Comment{$\hat{X}_b^{(i)} \in \mathbb{R}^{n_i \times (d_0 + 1)}$}
        \State Initialize $\hat{W}_a = \{\hat{W}_a^{(1)}, \hat{W}_a^{(2)}, \cdots, \hat{W}_a^{(L)}\}$ 
        \Statex \Comment{$\hat{W}_a^{(i)} \in \mathbb{R}^{(d_{i-1} + 1) \times d_i}$}
        \State Initialize $\hat{W}_b = \{\hat{W}_b^{(1)}, \hat{W}_b^{(2)}, \cdots, \hat{W}_b^{(L)}\}$ 
        \Statex \Comment{$\hat{W}_b^{(i)} \in \mathbb{R}^{(d_{i-1} + 1) \times d_i}$}

        \For{$round := 1$ \textbf{to} $t$}
            \For{$i := 1$ \textbf{to} $N$}
                \State $Y_a^{(L)}, Y_b^{(L)}, X_a^*, X_b^*, Z_a^*, Z_b^* \gets$ \textbf{S2PMLP-TR-FP}$(L, D, \hat{X}_a^{(i)}, \hat{X}_b^{(i)}, \hat{W}_a, \hat{W}_b, \rho)$
                \State $\hat{W}_a, \hat{W}_b \gets$ \textbf{S2PMLP-TR-BP}$(L, D, Y_a^{(L)}, Y_b^{(L)}, Y^{(i)}, X_a^*, X_b^*, Z_a^*, Z_b^*, \rho)$
            \EndFor
        \EndFor        
        \State \Return $\hat{W}_a, \hat{W}_b$
    \end{algorithmic}
\end{breakablealgorithm}

\textbf{Overall Process:}
In S2PMLP-TR, the inputs are the number of network layers $L \ge 1$, the set of neuron counts $D=\{d_0, d_1, \cdots, d_L\}$, the inputs $X_a, X_b \in \mathbb{R}^{n\times d_0}$, the labels $Y \in \{0, 1\}^{n\times d_L}$, the batch size $1 \leq B \leq n$, the number of iterations $t \geq 1$, the learning rate $\eta > 0$, and the splitting parameter $\rho \geq 2$. The outputs are the model parameters held by both parties, $\hat{W}_a=\{\hat{W}_a^{(1)},\hat{W}_a^{(2)},\cdots, \hat{W}_a^{(L)}\}$ and $\hat{W}_b=\{\hat{W}_b^{(1)},\hat{W}_b^{(2)},\cdots, \hat{W}_b^{(L)}\}$, where $\hat{W}_a^{(i)}, \hat{W}_b^{(i)} \in\mathbb{R}^{(d_{i-1} + 1)\times d_i} (1 \le i \le L)$. 

Alice and Bob first perform some plaintext operations to obtain several public variables: they compute $N = \left\lceil \frac{n}{B} \right\rceil$, which indicates the number of batches, determine the number of samples in the $i$-th batch $n_i \ (1 \le i \le N)$, and split the labels $Y$ into $N$ parts, where $Y^{(i)}$ represents the $i$-th part. Once these operations are completed, Alice and Bob preprocess their private data by concatenating columns, such that $\hat{X}_a = [1, X_a]$ and $\hat{X}_b = [0, X_b]$. This preprocessing is denoted as $\hat{X}_a, \hat{X}_b = \text{addcol}(X_a, X_b)$. Alice splits $\hat{X}_a$ into $N$ parts sequentially, where $\hat{X}_a^{(i)}$ represents the $i$-th part, and Bob performs the same operation. 

Alice and Bob then choose a parameter initialization method and initialize their respective network parameters $\hat{W}_a=\{\hat{W}_a^{(1)},\hat{W}_a^{(2)},\cdots, \hat{W}_a^{(L)}\}$ and $\hat{W}_b=\{\hat{W}_b^{(1)},\hat{W}_b^{(2)},\cdots, \hat{W}_b^{(L)}\}$. After completing the above preparations, the next step is to update the parameters using gradient descent. Each training iteration processes all batches sequentially. For each batch, Alice and Bob sequentially invoke S2PMLP-TR-FP (privacy-preserving forward propagation) and S2PMLP-TR-BP (privacy-preserving backward propagation) to complete the training for that batch. The training process continues until $t$ iterations are completed. The detailed process is described in Algorithm \ref{alg:S2PMLP-TR}.

\begin{breakablealgorithm}
    \caption{S2PMLP-TR Forward Propagation (S2PMLP-FP)}
    \label{alg:S2PMLP-TR-FP}
    \begin{algorithmic}[1]
        \Require $L$, $D$, $X_a, X_b$, $\hat{W}_a, \hat{W}_b$, $\rho$
        \Ensure $Y_a^{(L)}, Y_b^{(L)}$, $X_a^*$, $X_b^*$, $Z_a^*$, $Z_b^*$

        \State $Z_a^{(0)} = X_a$, $Z_b^{(0)} = X_b$ \Comment{$Z_a^{(0)}, Z_b^{(0)} \in \mathbb{R}^{n \times (d_0 + 1)}$}
        
        \For{$l := 1$ \textbf{to} $L - 1$}
            \State $X_a^{(l)}, X_b^{(l)} \gets$ \textbf{S2PHM}$((Z_a^{(l-1)}, \hat{W}_a^{(l)}), (Z_b^{(l-1)}, \hat{W}_b^{(l)}))$
            
            \Statex \Comment{$X_a^{(l)}, X_b^{(l)} \in \mathbb{R}^{n \times d_l}$}
            
            \State $Y_a^{(l)}, Y_b^{(l)} \gets$ \textbf{S2PRL}$(X_a^{(l)}, X_b^{(l)}, \rho)$
            \Comment{$Y_a^{(l)}, Y_b^{(l)} \in \mathbb{R}^{n \times d_l}$}
            
            \State $Z_a^{(l)}, Z_b^{(l)} = addcol(Y_a^{(l)}, Y_b^{(l)})$
            \Comment{$Z_a^{(l)}, Z_b^{(l)} \in \mathbb{R}^{n \times (d_l + 1)}$}
        \EndFor
        
        \State $X_a^{(L)}, X_b^{(L)} \gets$ \textbf{S2PHM}$((Z_a^{(L-1)}, \hat{W}_a^{(L)}), (Z_b^{(L-1)}, \hat{W}_b^{(L)}))$
        
        \Statex \Comment{$X_a^{(L)}, X_b^{(L)} \in \mathbb{R}^{n \times d_L}$}
        
        \State $Y_a^{(L)}, Y_b^{(L)} \gets$ \textbf{S2PSM}$(X_a^{(L)}, X_b^{(L)}, \rho)$
        \Comment{$Y_a^{(L)}, Y_b^{(L)} \in \mathbb{R}^{n \times d_L}$}
        
        \State \Return $Y_a^{(L)}$, $Y_b^{(L)}$, $X_a^*$, $X_b^*$, $Z_a^*$, $Z_b^*$
    \end{algorithmic}
\end{breakablealgorithm}

\textbf{Forward Propagation}: In the forward propagation of S2PMLP-TR, the inputs are the number of network layers $L \geq 1$, the set of neuron counts $D=\{d_0, d_1, \cdots, d_L\}$, the batch inputs held by both parties $X_a, X_b \in \mathbb{R}^{n\times (d_0 + 1)}$, the model parameters held by both parties $\hat{W}_a=\{\hat{W}_a^{(1)}, \hat{W}_a^{(2)}, \cdots, \hat{W}_a^{(L)}\}$ and $\hat{W}_b=\{\hat{W}_b^{(1)}, \hat{W}_b^{(2)}, \cdots, \hat{W}_b^{(L)}\}$, and the splitting parameter $\rho \geq 2$. The outputs are the outputs of the model’s output layer held by both parties $Y_a^{(L)}, Y_b^{(L)} \in \mathbb{R}^{n\times d_L}$, as well as the intermediate variables for all layers: $X_a^* = \{X_a^{(0)}, X_a^{(1)}, \cdots, X_a^{(L)}\}$, $X_b^* = \{X_b^{(0)}, X_b^{(1)}, \cdots, X_b^{(L)}\}$, $Z_a^* = \{Z_a^{(0)}, Z_a^{(1)}, \cdots, Z_a^{(L)}\}$, and $Z_b^* = \{Z_b^{(0)}, Z_b^{(1)}, \cdots, Z_b^{(L)}\}$. 

First, set $Z_a^{(0)} = X_a$ and $Z_b^{(0)} = X_b$. Then, iterate through the first $L-1$ layers. For the $l$-th layer $(1 \leq l < L)$: Alice and Bob compute $(Z_a^{(l-1)} + Z_b^{(l-1)}) \times (\hat{W}_a^{(l)} + \hat{W}_b^{(l)}) = X_a^{(l)} + X_b^{(l)}$ securely using the S2PHM (secure two-party matrix multiplication) protocol. They then compute $\text{relu}(X_a^{(l)} + X_b^{(l)}) = Y_a^{(l)} + Y_b^{(l)}$ securely using the S2PRL (secure two-party ReLU activation) protocol. Finally, Alice and Bob preprocess $Y_a^{(l)}$ and $Y_b^{(l)}$ by concatenating columns to add a constant, denoted as $Z_a^{(l)}, Z_b^{(l)} = \text{addcol}(Y_a^{(l)}, Y_b^{(l)})$.

For the $L$-th layer: Alice and Bob compute $(Z_a^{(L-1)} + Z_b^{(L-1)}) \times (\hat{W}_a^{(L)} + \hat{W}_b^{(L)}) = X_a^{(L)} + X_b^{(L)}$ securely using the S2PHM protocol. They then compute $\text{softmax}(X_a^{(L)} + X_b^{(L)}) = Y_a^{(L)} + Y_b^{(L)}$ securely using the S2PSM (secure two-party Softmax computation) protocol. The detailed process of forward propagation is shown in Algorithm \ref{alg:S2PMLP-TR-FP}.

\textbf{Backward Propagation}: In the backward propagation of S2PMLP-TR, the inputs are the number of network layers $L \geq 1$, the set of neuron counts $D=\{d_0, d_1, \cdots, d_L\}$, the batch inputs held by both parties $X_a, X_b \in \mathbb{R}^{n\times (d_0 + 1)}$, the outputs of the model’s output layer held by both parties $Y_a^{(L)}, Y_b^{(L)} \in \mathbb{R}^{n\times d_L}$, the batch labels $Y \in \mathbb{R}^{n\times d_L}$, the intermediate variables for all layers: $X_a^* = \{X_a^{(0)}, X_a^{(1)}, \cdots, X_a^{(L)}\}$, $X_b^* = \{X_b^{(0)}, X_b^{(1)}, \cdots, X_b^{(L)}\}$, $Z_a^* = \{Z_a^{(0)}, Z_a^{(1)}, \cdots, Z_a^{(L)}\}$, $Z_b^* = \{Z_b^{(0)}, Z_b^{(1)}, \cdots, Z_b^{(L)}\}$, the model parameters held by both parties $\hat{W}_a=\{\hat{W}_a^{(1)}, \hat{W}_a^{(2)}, \cdots, \hat{W}_a^{(L)}\}$ and $\hat{W}_b=\{\hat{W}_b^{(1)}, \hat{W}_b^{(2)}, \cdots, \hat{W}_b^{(L)}\}$, the learning rate $\eta > 0$, and the splitting parameter $\rho \geq 2$. The outputs are the updated model parameters held by both parties $\hat{W}_a=\{\hat{W}_a^{(1)}, \hat{W}_a^{(2)}, \cdots, \hat{W}_a^{(L)}\}$ and $\hat{W}_b=\{\hat{W}_b^{(1)}, \hat{W}_b^{(2)}, \cdots, \hat{W}_b^{(L)}\}$.

\begin{breakablealgorithm}
    \caption{S2PMLP-TR Backward Propagation (S2PMLP-BP)}
    \label{alg:S2PMLP-TR-BP}
    \begin{algorithmic}[1]
        \Require $L$, $D$, $Y_a^{(L)}, Y_b^{(L)}, Y$, $X_a^*$, $X_b^*$, $Z_a^*$, $Z_b^*$, $\hat{W}_a$, $\hat{W}_b$, $\eta$, $\rho$
        \Ensure $\hat{W}_a$, $\hat{W}_b$

        \State $G_a^{(L)} = Y_a^{(L)} - Y$, $G_b^{(L)} = Y_b^{(L)}$ 
        \Comment{$G_a^{(L)}, G_b^{(L)} \in \mathbb{R}^{n \times d_L}$}
        
        \For{$l := L-1$ \textbf{to} $1$}
            \State $G_a^{(l)}, G_b^{(l)} \gets$ \textbf{S2PG}$(\hat{W}_a^{(l + 1)}, \hat{W}_b^{(l + 1)}, G_a^{(l+1)},G_b^{(l+1)}, X_a^{(l)}, $ $ X_b^{(l)}, \rho)$

        \EndFor

        \For{$l := 1$ \textbf{to} $L$}
            \State $T_a = Z_a^{(l-1)^T}$, $T_b = Z_b^{(l-1)^T}$ 
            \Comment{$T_a, T_b \in \mathbb{R}^{(d_{l-1} + 1) \times n}$}
            
            \State $\delta W_a^{(l)}, \delta W_b^{(l)} \gets$ \textbf{S2PHM}$((T_a, G_a^{(l)}), (T_b, G_b^{(l)}))$ 
            
            \Statex \Comment{$\delta W_a^{(l)}, \delta W_b^{(l)} \in \mathbb{R}^{(d_{l-1} + 1) \times d_l}$}
            
            \State $\hat{W}_a^{(l)} := \hat{W}_a^{(l)} - \eta \cdot \delta W_a^{(l)}$ 
            \Comment{$\hat{W}_a^{(l)} \in \mathbb{R}^{(d_{l-1} + 1) \times d_l}$}
            
            \State $\hat{W}_b^{(l)} := \hat{W}_b^{(l)} - \eta \cdot \delta W_b^{(l)}$ 
            \Comment{$\hat{W}_b^{(l)} \in \mathbb{R}^{(d_{l-1} + 1) \times d_l}$}
        \EndFor

        \State \Return $\hat{W}_a$, $\hat{W}_b$
    \end{algorithmic}
\end{breakablealgorithm}

Alice and Bob first compute the gradient for the $L$-th layer as $G_a^{(L)} = Y_a^{(L)} - Y$ and $G_b^{(L)} = Y_b^{(L)}$. They then proceed with backpropagation from layer $L-1$ to layer 1. For each layer, Alice and Bob jointly invoke the S2PG-MLP protocol to compute the gradients $G_a^{(l)}, G_b^{(l)}$ for that layer. 

Once the gradient for each layer is computed, the model parameters are updated. For the $l$-th layer $(1 \le l \le L)$: Alice and Bob first compute $T_a = Z_a^{(l-1)^T}$ and $T_b = Z_b^{(l-1)^T}$, the transposes of the intermediate variables from the previous layer. They then securely compute $(T_a + T_b) \times (G_a^{(l)} + G_b^{(l)}) = \delta W_a^{(l)} + \delta W_b^{(l)}$ using the S2PHM (secure two-party matrix multiplication) protocol. Finally, they each update their local model parameters: $\hat{W}_a^{(l)} := \hat{W}_a^{(l)} - \eta \cdot \delta W_a^{(l)}$ and $\hat{W}_b^{(l)} := \hat{W}_b^{(l)} - \eta \cdot \delta W_b^{(l)}$. The detailed process of backward propagation is shown in Algorithm \ref{alg:S2PMLP-TR-BP}.

\textbf{Correctness}: For simplicity, we assume that each training epoch uses the entire training dataset (i.e., $B = n$). Let $D = \{d_0, d_1, \cdots, d_L\}$ represent the network structure. During forward propagation in each epoch, for the $l$-th layer $(1 \le l < L)$: $X_a^{(l)} + X_b^{(l)} = (Z_a^{(l-1)} + Z_b^{(l-1)}) \times (\hat{W}_a^{(l)} + \hat{W}_b^{(l)}) = Z^{(l-1)} \times \hat{W}^{(l)} = X^{(l)}$, and $Y_a^{(l)} + Y_b^{(l)} = \mathrm{relu}(X_a^{(l)} + X_b^{(l)}) = \mathrm{relu}(X^{(l)}) = Y^{(l)}$. For the $L$-th layer: $X_a^{(L)} + X_b^{(L)} = (Z_a^{(L-1)} + Z_b^{(L-1)}) \times (\hat{W}_a^{(L)} + \hat{W}_b^{(L)}) = Z^{(L-1)} \times \hat{W}^{(L)} = X^{(L)}$, and $Y_a^{(L)} + Y_b^{(L)} = \mathrm{softmax}(X_a^{(L)} + X_b^{(L)}) = \mathrm{softmax}(X^{(L)}) = Y^{(L)}$.

These results conform to the centralized forward propagation formulas, thus proving the correctness of forward propagation.

During backpropagation in each epoch, for the $L$-th layer: $G_a^{(L)} + G_b^{(L)} = Y_a^{(L)} + Y_b^{(L)} - Y = Y^{(L)} - Y = G^{(L)}$.

For the $l$-th layer $(1 \le l < L)$, based on the correctness of S2PG-MLP, the gradient descent in S2PMLP-TR adheres to the centralized gradient propagation formula. For parameter updates in the $l$-th layer $(1 \le l \le L)$: $\hat{W}_a^{(l)} + \hat{W}_b^{(l)} := \hat{W}_a^{(l)} + \hat{W}_b^{(l)} - \eta \cdot (\delta W_a^{(l)} + \delta W_b^{(l)})$, which is equivalent to: $\hat{W}^{(l)} := \hat{W}^{(l)} - \eta \cdot \delta W^{(l)}$. This satisfies the centralized parameter update formula, thus proving the correctness of backpropagation. Based on the correctness of both forward and backpropagation, we can verify the overall training correctness of S2PMLP-TR.

We analyze the probability of anomaly detection failure for S2PMLP-TR in each batch. Assuming the network structure has $L$ layers ($L \ge 1$), S2PMLP-TR will only fail to detect an anomaly if all the sub-protocols S2PHM, S2RL, S2PG-MLP, and S2PSM fail to detect the anomaly. Therefore, the probability of S2PMLP-TR failing to detect a computational anomaly in each batch is $P_f(S2PMLP\text{-}TR) = P_f(S2PHM)^{2L} \cdot P_f(S2PRL)^{L-1} \cdot P_f(S2PG\text{-}MLP)^{L-1} \cdot P_f(S2PSM) \leq \left(\frac{1}{4^l}\right)^{8L+2} \leq \left(\frac{1}{4^l}\right)^{10} \approx 3.87 \times 10^{-121}$ $(l=20)$.

We can prove the security of S2PMLP-TR based on the UC framework.

\begin{theorem}\label{theorem:Secure-S2PMLP-TR}
    The S2PMLP-TR protocol is secure under the semi-honest adversarial model.
\end{theorem}

\subsection{Secure Two-Party MLP Inference Model}

The problem definition of S2PMLP-PR is as follows:

\begin{problem}[Secure Two-Party Multi-Layer Perceptron Inference]\label{Problem-S2PMLP-PR}
    In the context of heterogeneous distributed data, let the network parameters be given by $D = \{d_0, d_1, \cdots, d_L\}$, where each element represents the number of neurons in each layer, including the input, hidden, and output layers, and $L$ represents the number of network layers excluding the input layer. Alice and Bob respectively hold parts of the prediction dataset, denoted as $X_a$ and $X_b$, as well as the model parameters for all layers: $\hat{W}_a = \{\hat{W}_a^{(1)}, \hat{W}_a^{(2)}, \cdots, \hat{W}_a^{(L)}\}$ and $\hat{W}_b = \{\hat{W}_b^{(1)}, \hat{W}_b^{(2)}, \cdots, \hat{W}_b^{(L)}\}$. Their goal is to perform privacy-preserving inference of the multi-layer perceptron model such that Alice obtains $Y_a^{(L)}$ and Bob obtains $Y_b^{(L)}$, satisfying: $Y_a^{(L)} + Y_b^{(L)} = Y^{(L)}$, where $Y^{(L)}$ represents the true prediction result of the model.
\end{problem}

\textbf{S2PMLP-PR Protocol Description:}  
In S2PMLP-PR, the input includes the number of network layers $L \geq 1$, the set of neuron counts $D = \{d_0, d_1, \cdots, d_L\}$, the input data for both parties $X_a, X_b \in \mathbb{R}^{n \times d_0}$, the model parameters held by both parties $\hat{W}_a = \{\hat{W}_a^{(1)}, \hat{W}_a^{(2)}, \cdots, \hat{W}_a^{(L)}\}$ and $\hat{W}_b = \{\hat{W}_b^{(1)}, \hat{W}_b^{(2)}, \cdots, \hat{W}_b^{(L)}\}$, and the splitting parameter $\rho \geq 2$. The output is the model output of the output layer held by both parties, $Y_a^{(L)}, Y_b^{(L)} \in \mathbb{R}^{n \times d_L}$.

The S2PMLP-PR protocol follows the same forward propagation flow as S2PMLP-TR. Specifically, Alice and Bob preprocess their private data by concatenating columns such that $Z_a^{(0)} = [1, X_a]$ and $Z_b^{(0)} = [0, X_b]$. This preprocessing step is denoted as: $Z_a^{(0)}, Z_b^{(0)} = \text{addcol}(X_a, X_b)$.

Then, for each of the first $L-1$ layers, Alice and Bob jointly compute: $(Z_a^{(l-1)} + Z_b^{(l-1)}) \times (\hat{W}_a^{(l)} + \hat{W}_b^{(l)}) = X_a^{(l)} + X_b^{(l)}$, using the S2PHM protocol. They then jointly compute: $\mathrm{relu}(X_a^{(l)} + X_b^{(l)}) = Y_a^{(l)} + Y_b^{(l)}$, using the S2PRL protocol. Finally, they perform a column concatenation step: $Z_a^{(l)}, Z_b^{(l)} = \text{addcol}(Y_a^{(l)}, Y_b^{(l)})$. 

For the $L$-th layer, Alice and Bob first jointly compute: $(Z_a^{(L-1)} + Z_b^{(L-1)}) \times (\hat{W}_a^{(L)} + \hat{W}_b^{(L)}) = X_a^{(L)} + X_b^{(L)}$, using the S2PHM protocol. They then jointly compute: $\mathrm{softmax}(X_a^{(L)} + X_b^{(L)}) = Y_a^{(L)} + Y_b^{(L)}$, using the S2PSM protocol. The detailed process is shown in Algorithm~\ref{alg:S2PMLP-PR}.

\begin{breakablealgorithm}
    \caption{S2PMLP-PR}
    \label{alg:S2PMLP-PR}
    \begin{algorithmic}[1]
        \Require $L$, $D$, $X_a$, $X_b$, $\hat{W}_a$, $\hat{W}_b$, $\rho$
        \Ensure $Y_a^{(L)}$, $Y_b^{(L)}$, $X_a^*$, $X_b^*$, $Z_a^*$, $Z_b^*$

        \State $Z_a^{(0)}, Z_b^{(0)} = addcol(X_a, X_b)$ 
        \Comment{$Z_a^{(0)}, Z_b^{(0)} \in \mathbb{R}^{n \times (d_0 + 1)}$}

        \For{$l := 1$ \textbf{to} $L - 1$}
            \State $X_a^{(l)}, X_b^{(l)} \gets$ \textbf{S2PHM}$((Z_a^{(l-1)}, \hat{W}_a^{(l)}), (Z_b^{(l-1)}, \hat{W}_b^{(l)}))$
            
            \Statex \Comment{$X_a^{(l)}, X_b^{(l)} \in \mathbb{R}^{n \times d_l}$}
            
            \State $Y_a^{(l)}, Y_b^{(l)} \gets$ \textbf{S2PRL}$\left(X_a^{(l)}, X_b^{(l)}, \rho\right)$
            \Comment{$Y_a^{(l)}, Y_b^{(l)} \in \mathbb{R}^{n \times d_l}$}
            
            \State $Z_a^{(l)}, Z_b^{(l)} = addcol(Y_a^{(l)}, Y_b^{(l)})$
            \Comment{$Z_a^{(l)}, Z_b^{(l)} \in \mathbb{R}^{n \times (d_l + 1)}$}
        \EndFor

        \State $X_a^{(L)}, X_b^{(L)} \gets$ \textbf{S2PHM}$((Z_a^{(L-1)}, \hat{W}_a^{(L)}), (Z_b^{(L-1)}, \hat{W}_b^{(L)}))$
        
        \Statex \Comment{$X_a^{(L)}, X_b^{(L)} \in \mathbb{R}^{n \times d_L}$}

        \State $Y_a^{(L)}, Y_b^{(L)} \gets$ \textbf{S2PSM}$(X_a^{(L)}, X_b^{(L)}, \rho)$
        \Comment{$Y_a^{(L)}, Y_b^{(L)} \in \mathbb{R}^{n \times d_L}$}

        \State \Return $Y_a^{(L)}$, $Y_b^{(L)}$, $X_a^*$, $X_b^*$, $Z_a^*$, $Z_b^*$
    \end{algorithmic}
\end{breakablealgorithm}

\textbf{Correctness}: Let $D = \{d_0, d_1, \cdots, d_L\}$ represent the network structure. During the prediction process of the multi-layer perceptron, for the $l$-th layer $(1 \le l < L)$, we have $X_a^{(l)} + X_b^{(l)} = (Z_a^{(l-1)} + Z_b^{(l-1)}) \times (\hat{W}_a^{(l)} + \hat{W}_b^{(l)}) = Z^{(l-1)} \times \hat{W}^{(l)} = X^{(l)}$, and $Y_a^{(l)} + Y_b^{(l)} = \mathrm{relu}(X_a^{(l)} + X_b^{(l)}) = \mathrm{relu}(X^{(l)}) = Y^{(l)}$. 

For the $L$-th layer, we have $X_a^{(L)} + X_b^{(L)} = (Z_a^{(L-1)} + Z_b^{(L-1)}) \times (\hat{W}_a^{(L)} + \hat{W}_b^{(L)}) = Z^{(L-1)} \times \hat{W}^{(L)} = X^{(L)}$, and $Y_a^{(L)} + Y_b^{(L)} = \mathrm{softmax}(X_a^{(L)} + X_b^{(L)}) = \mathrm{softmax}(X^{(L)}) = Y^{(L)}$. 

These results conform to the centralized model prediction formulas, thus proving the correctness of S2PMLP-PR.

Assuming the network structure has $L$ layers ($L \geq 1$), S2PMLP-PR will only fail to detect an anomaly if all the sub-protocols S2PHM, S2RL, and S2PSM fail to detect the anomaly. Therefore, the probability of S2PMLP-PR failing to detect a computational anomaly is given by $P_f(S2PMLP\text{-}PR) = P_f(S2PHM)^{L} \cdot P_f(S2PRL)^{L-1} \cdot P_f(S2PSM) \leq \left(\frac{1}{4^l}\right)^{3L+5} \leq \left(\frac{1}{4^l}\right)^{8} \approx 4.68 \times 10^{-97} \quad (l=20)$.

We can prove the security of S2PMLP-PR based on the UC framework.

\begin{theorem}\label{theorem:Secure-S2PMLP-PR}
    The S2PMLP-PR protocol is secure under the semi-honest adversarial model.
\end{theorem}

\section{Theoretical Complexity Analysis}\label{Theoretical}
In this section, we present a comprehensive theoretical analysis of the computational and communication complexity of the protocols proposed in this chapter. Each protocol is analyzed with respect to three stages: 	extit{preprocessing}, 	extit{online computation}, and 	extit{verification}. We also provide detailed breakdowns of both the total communication cost and the number of communication rounds.

We specifically analyze the complexity of the following protocols: S2PHP, S2PSCR, S2PHHP, S2PDRL, S2PRL, S2PSM, S2PMLP-TR, and S2PMLP-PR. For simplification, we assume that the input of all basic protocols is a square matrix of size $n \times n$. For MLP protocols, we assume the neural network has $L$ hidden layers (with $L \ge 2$), each consisting of $d$ neurons, and the input contains $n$ samples. Bias terms are omitted for brevity. The complexity analysis is given per training round on the entire dataset. We also denote $\rho$ as the partition parameter and $\ell$ as the bit-length of each data element.

\subsection{Computational Complexity}
\noindent The results are summarized in \autoref{tab:Computational Complexity}, detailing the cost of each protocol across all stages.

\begin{table}[htbp]
  \centering
  \caption{Theoretical Computational Complexity of Each Protocol}
  \vspace{1mm}
    \resizebox{\linewidth}{!}{
        \begin{tabular}{c c c c c c c c c}
        \toprule
        \multirow{2}[3]{*}{\textbf{Protocol}} & & \multicolumn{5}{c}{\textbf{Stage-wise Computational Complexity}} & & \multirow{2}[3]{*}{\textbf{Total}} \\
        \cmidrule{3-7}
         & & \textbf{Preprocessing} & & \textbf{Online} & & \textbf{Verification} & & \\
        \midrule
        \textbf{S2PHP} & & $\mathcal{O}(n^4\rho^2)$ & & $\mathcal{O}(n^4\rho^2)$ & & $\mathcal{O}(n^4)$ & & $\mathcal{O}(n^4\rho^2)$ \\
        \textbf{S2PSCR} & & $\mathcal{O}(n^4\rho^2)$ & & $\mathcal{O}(n^4\rho^2)$ & & $\mathcal{O}(n^4)$ & & $\mathcal{O}(n^4\rho^2)$ \\
        \textbf{S2PHHP} & & $\mathcal{O}(n^4\rho^2)$ & & $\mathcal{O}(n^4\rho^2)$ & & $\mathcal{O}(n^4)$ & & $\mathcal{O}(n^4\rho^2)$ \\
        \textbf{S2PDRL} & & $\mathcal{O}(n^4\rho)$ & & $\mathcal{O}(n^4\rho)$ & & $\mathcal{O}(n^4)$ & & $\mathcal{O}(n^4\rho)$ \\
        \textbf{S2PRL} & & $\mathcal{O}(n^4\rho)$ & & $\mathcal{O}(n^4\rho)$ & & $\mathcal{O}(n^4)$ & & $\mathcal{O}(n^4\rho)$ \\
        \textbf{S2PSM} & & $\mathcal{O}(n^4\rho^2)$ & & $\mathcal{O}(n^4\rho^2)$ & & $\mathcal{O}(n^4)$ & & $\mathcal{O}(n^4\rho^2)$ \\
        \textbf{S2PMLP-TR} & & $\mathcal{O}(Ln^2d^2\rho + n^2d^2\rho^2)$ & & $\mathcal{O}(Ln^2d^2\rho + n^2d^2\rho^2)$ & & $\mathcal{O}(Ln^2d^2)$ & & $\mathcal{O}(Ln^2d^2\rho + n^2d^2\rho^2)$ \\
        \textbf{S2PMLP-PR} & & $\mathcal{O}(Ln^2d^2\rho + n^2d^2\rho^2)$ & & $\mathcal{O}(Ln^2d^2\rho + n^2d^2\rho^2)$ & & $\mathcal{O}(Ln^2d^2)$ & & $\mathcal{O}(Ln^2d^2\rho + n^2d^2\rho^2)$ \\
        \bottomrule
        \end{tabular}}
  \label{tab:Computational Complexity}
\end{table}

\noindent
Take \textbf{S2PHP} as an example. In the preprocessing phase, the $n \times n$ matrix is partitioned into two matrices of dimension $n^2 \times \rho^2$ and $\rho^2 \times n^2$, respectively. The dominant cost arises from the S2PM protocol invoked internally, which performs a matrix multiplication of these sizes. Hence, the preprocessing and online computation both incur $\mathcal{O}(n^4\rho^2)$ operations. The verification step checks the integrity of an $n^2 \times n^2$ matrix using methods in S2PM, resulting in $\mathcal{O}(n^4)$ operations. The complexity of other protocols follows similar logic.

\subsection{Communication Complexity}
\noindent
The communication complexity of each protocol, including round count and total bits transferred, is listed in \autoref{tab:Communication Complexity}.

\begin{table}[htbp]
  \centering
  \caption{Theoretical Communication Complexity of Each Protocol}
  \vspace{1mm}
  \resizebox{\linewidth}{!}{
    \begin{tabular}{c c c}
      \toprule
      \textbf{Protocol} & \textbf{Rounds} & \textbf{Communication Cost [bits]} \\
      \midrule
      \textbf{S2PHP}   & 6        & $(7n^4 + 4n^2\rho^2)\ell$ \\
      \textbf{S2PSCR}   & 19       & $(21n^4 + 12n^2\rho^2 + n^2)\ell$ \\
      \textbf{S2PHHP}  & 12       & $(14n^4 + 8n^2\rho^2)\ell$ \\
      \textbf{S2PDRL}  & 8        & $(7n^4 + 8n^2\rho + 2n^2)\ell$ \\
      \textbf{S2PRL}   & 8        & $(7n^4 + 8n^2\rho + 2n^2)\ell$ \\
      \textbf{S2PSM}   & 37       & $(21n^4 + 12n^2\rho^2 + 21n^2 + 12n\rho^2 + n)\ell$ \\
      \textbf{S2PMLP-TR} & $52L+9$  & $((14L+7)n^2d^2 + 12nd\rho^2 + (16L-16)nd\rho + 21n^2 + 12n\rho^2 + (40L-14)nd + (22L-4)d^2 + n)\ell$ \\
      \textbf{S2PMLP-PR} & $20L+29$ & $((7L+14)n^2d^2 + 12nd\rho^2 + (8L-8)nd\rho + 21n^2 + 12n\rho^2 + (10L-2)nd + 14Ld^2 + n)\ell$ \\
      \bottomrule
    \end{tabular}
  }
  \label{tab:Communication Complexity}
\end{table}

\noindent
As an example, the S2PHP protocol invokes the S2PM primitive once and incurs 6 communication rounds. The total communication cost is determined by the size of input matrices $n^2 \times \rho^2$ and $\rho^2 \times n^2$, yielding $(7n^4 + 4n^2\rho^2)\ell$ bits. Other protocol costs follow from their constituent primitives and message sizes.

\section{PERFORMANCE EVALUATION}\label{Experiments}
\subsection{Setup Configuration}
All protocols in the EVA-S2PMLP framework are implemented in Python as independent modules. Experiments are conducted on a local machine equipped with an Intel$^\circledR$ Core\texttrademark~Ultra 7 processor (22 cores), 32 GB RAM, running Ubuntu 22.04 LTS. To minimize discrepancies caused by variable network conditions, we simulate both LAN and WAN environments during testing. Specifically, LAN is configured with 10.1 Gbps bandwidth and 0.1 ms latency, while WAN is set with 300 Mbps bandwidth and 40 ms latency.
The system involves four simulated computation nodes launched via different ports: two act as data owners, one acts as a semi-honest commodity server (CS) responsible only for generating randomness during the offline phase, and one serves as the client who initiates the computation and retrieves the final result.

\subsection{Parameter Description}
We evaluate the performance of four fundamental protocols (S2PHP, S2PSCR, S2PRL, S2PSM) in terms of computation cost, communication overhead, and precision loss. Since most existing secure logistic regression frameworks use approximate methods and do not require multiplicative-to-additive conversions or reciprocal operations (as in our S2PLoR approach), we only include comparisons for low-level multiplication protocol S2PVEM and activation function protocol S2PVS.
For S2PHP, we compare our results under the semi-honest setting against four representative frameworks: Crypten, SecretFlow, LibOTe, and TenSEAL. For S2PSCR, we consider Crypten, SecretFlow, LibOTe, and Rosetta. For S2PRL, we include Crypten, LibOTe, Rosetta, and MP-SPDZ. For S2PSM, we evaluate against Crypten, LibOTe, and MP-SPDZ.
Since many existing frameworks operate over ring $\mathbb{Z}_{2^k}$ and suffer from floating-point limitations, we generate random inputs using 16-digit significands formatted as $1.a_1a_2\cdots a_{15}\times 10^\delta$ with $\delta \in [-x, x]$, where $x \in \mathbb{Z}$. In runtime evaluations, we set $x = 4$; in precision evaluations, we sweep $x$ from 0 to 8 with a step size of 2, covering five representative precision ranges.

\subsection{Efficiency of Basic Protocols}
We first assess the runtime efficiency of the four basic protocols in the EVA-S2PMLP framework. Average results are summarized in \autoref{tab:Our-Time}. We then compare their communication cost across existing frameworks, as shown in \autoref{fig: Communication Comparison}. Total runtime comparisons under LAN and WAN settings are depicted in \autoref{fig: LAN Time Comparison} and \autoref{fig: WAN Time Comparison}, respectively.
In addition to the cross-framework comparisons, we evaluate the verification overhead of all protocols by measuring both the number of verification rounds and the proportion of total time spent on verification as the input size scales, visualized in \autoref{fig: Verification Rate}. Furthermore, we examine the scalability of S2PHP and S2PSCR with large-scale matrix inputs, with results presented in \autoref{fig: Huge Matrix}.

\begin{table}[htbp]
  \centering
  \caption{Performance Evaluation of Basic Protocols in EVA-S2PMLP}
  \vspace{1mm}
  \resizebox{\linewidth}{!}{
    \begin{tabular}{cccccccccccc}
      \toprule
      \multicolumn{1}{c}{\multirow{2}[2]{*}{\textbf{Protocol}}} & 
      \multicolumn{1}{c}{\multirow{2}[2]{*}{\textbf{Dimension}}} & 
      \multicolumn{2}{c}{\textbf{Communication Overhead}} & 
      \multicolumn{4}{c}{\textbf{Computation Time (S)}} & 
      \multicolumn{2}{c}{\textbf{Total Time (S)}} \\
      \cmidrule(lr){3-4} \cmidrule(lr){5-8} \cmidrule(lr){9-10}
       &  &  \textbf{Volume (KB)} &  \textbf{Rounds} & \textbf{Preprocessing} & \textbf{Online Computation} & \textbf{Verification} & \textbf{Total Computation} & \textbf{LAN} & \textbf{WAN} \\
      \midrule
      \multirow{5}{*}{\textbf{S2PHP}} 
        & 10 & 188.14  & \multirow{5}{*}{6}  & 7.89E-04 & 1.65E-04 & 1.11E-04 & 1.07E-03 & 1.81E-03 & 2.46E-01 \\
        & 20 & 756.17  &  & 2.23E-03 & 5.38E-04 & 2.77E-04 & 3.05E-03 & 4.22E-03 & 2.63E-01 \\
        & 30 & 1724.92 &  & 4.40E-03 & 1.16E-03 & 5.49E-04 & 6.11E-03 & 8.02E-03 & 2.91E-01 \\
        & 40 & 3074.14 &  & 7.46E-03 & 2.10E-03 & 9.40E-04 & 1.05E-02 & 1.34E-02 & 3.31E-01 \\
        & 50 & 4798.83 &  & 1.18E-02 & 3.24E-03 & 1.49E-03 & 1.65E-02 & 2.08E-02 & 3.82E-01 \\
      \midrule
      \multirow{5}{*}{\textbf{S2PSCR}} 
        & 10 & 565.34  & \multirow{5}{*}{19}  & 2.71E-03 & 5.90E-04 & 3.98E-04 & 3.70E-03 & 6.02E-03 & 7.78E-01 \\
        & 20 & 2271.78 &  & 8.41E-03 & 1.86E-03 & 1.02E-03 & 1.13E-02 & 1.49E-02 & 8.30E-01 \\
        & 30 & 5181.94 &  & 1.82E-02 & 4.09E-03 & 2.24E-03 & 2.46E-02 & 3.04E-02 & 9.20E-01 \\
        & 40 & 9235.06 &  & 2.36E-02 & 5.73E-03 & 2.89E-03 & 3.22E-02 & 4.11E-02 & 1.03E+00 \\
        & 50 & 14416.16&  & 3.46E-02 & 8.16E-03 & 4.42E-03 & 4.71E-02 & 5.99E-02 & 1.18E+00 \\
      \midrule
      \multirow{5}{*}{\textbf{S2PRL}} 
        & 10 & 195.45  & \multirow{5}{*}{10}  & 6.61E-04 & 2.36E-04 & 1.35E-04 & 1.03E-03 & 2.18E-03 & 4.06E-01 \\
        & 20 & 782.23  &  & 1.55E-03 & 7.87E-04 & 3.63E-04 & 2.70E-03 & 4.29E-03 & 4.23E-01 \\
        & 30 & 1782.23 &  & 3.08E-03 & 1.69E-03 & 7.58E-04 & 5.52E-03 & 7.87E-03 & 4.52E-01 \\
        & 40 & 3175.20 &  & 5.57E-03 & 2.62E-03 & 1.03E-03 & 9.22E-03 & 1.26E-02 & 4.92E-01 \\
        & 50 & 4956.14 &  & 6.37E-03 & 3.04E-03 & 1.50E-03 & 1.09E-02 & 1.57E-02 & 5.40E-01 \\
      \midrule
      \multirow{5}{*}{\textbf{S2PSM}} 
        & 10 & 1144.72 & \multirow{5}{*}{47}  & 4.86E-03 & 1.52E-03 & 6.92E-04 & 7.07E-03 & 1.26E-02 & 1.92E+00 \\
        & 20 & 4592.75 &  & 1.81E-02 & 4.20E-03 & 1.99E-03 & 2.43E-02 & 3.24E-02 & 2.02E+00 \\
        & 30 & 10471.66&  & 2.70E-02 & 6.04E-03 & 3.43E-03 & 3.65E-02 & 4.91E-02 & 2.19E+00 \\
        & 40 & 18659.94&  & 4.51E-02 & 1.03E-02 & 5.64E-03 & 6.11E-02 & 7.99E-02 & 2.43E+00 \\
        & 50 & 29127.59&  & 6.98E-02 & 1.54E-02 & 8.75E-03 & 9.39E-02 & 1.21E-01 & 2.73E+00 \\

        \midrule
      \multirow{5}{*}{\textbf{S2PG}} 
        & 10 & 212.38 & \multirow{5}{*}{20}  & 5.07E-04 & 3.76E-04 & 1.66E-04 & 1.05E-03 & 3.21E-03 & 8.07E-01 \\
        & 20 & 841.34 &  & 8.78E-04 & 7.84E-04 & 3.51E-04 & 2.01E-03 & 4.65E-03 & 8.24E-01 \\
        & 30 & 1911.66 &  & 1.84E-03 & 1.68E-03 & 9.52E-04 & 4.47E-03 & 7.92E-03 & 8.54E-01 \\
        & 40 & 3403.06 &  & 5.71E-03 & 2.86E-03 & 1.22E-03 & 9.79E-03 & 1.44E-02 & 8.98E-01 \\
        & 50 & 5310.56 &  & 4.47E-03 & 3.95E-03 & 1.84E-03 & 1.03E-02 & 1.63E-02 & 9.49E-01 \\
      \bottomrule
    \end{tabular}
  }
  \label{tab:Our-Time}
\end{table}

As shown in Table~\ref{tab:Our-Time}, despite the incorporation of verification mechanisms, our basic protocols still achieve fast execution under both LAN and WAN settings due to their streamlined computational workflows. Note that some protocols exhibit a slight discrepancy in communication rounds compared to the theoretical analysis in Table~\ref{tab:Communication Complexity}. This is because we have optimized the protocols for precision, which alters the number of communication rounds.

\begin{figure}[htbp]
    \centering
    \subfigure[Communication overhead of S2PHP]{
        \includegraphics[width=0.45\linewidth]{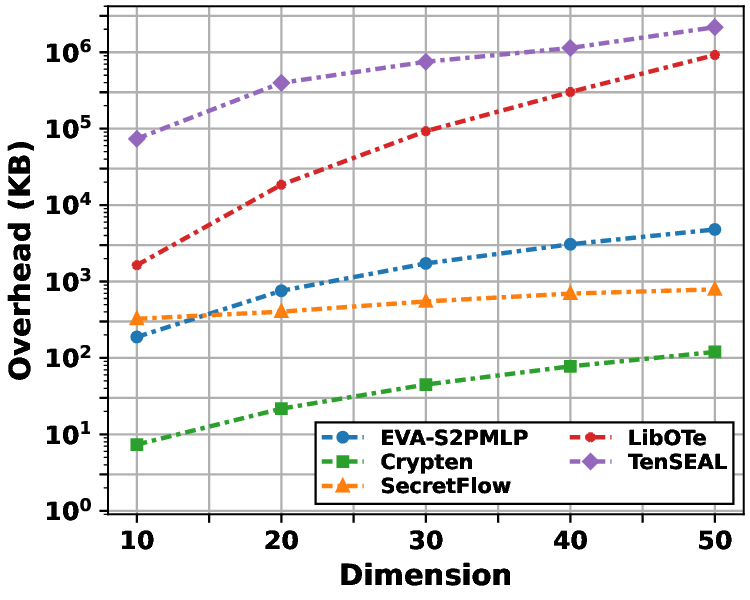}
        \Description{Communication overhead of S2PHP}
        \label{fig: S2PHP Communication}
    }
    \hspace{0.3cm}
    \subfigure[Communication overhead of S2PSCR]{
        \includegraphics[width=0.45\linewidth]{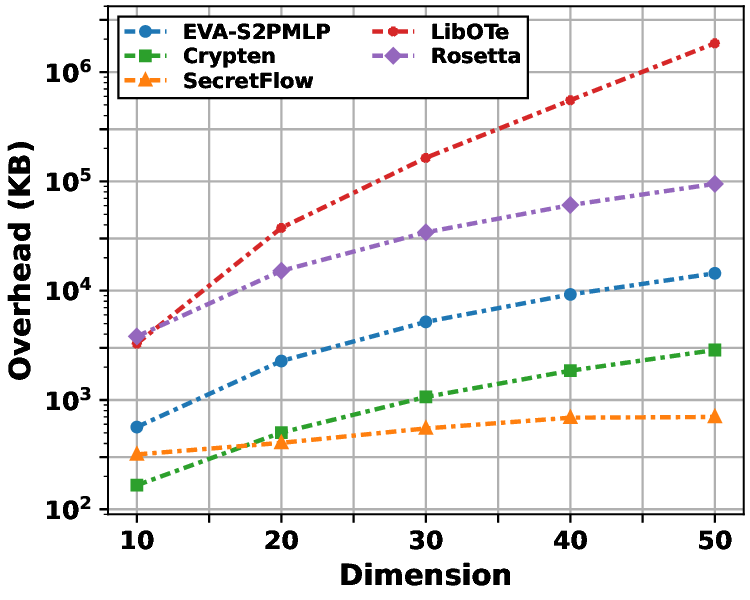}
        \Description{Communication overhead of S2PSCR}
        \label{fig: S2PSCR Communication}
    }
    \subfigure[Communication overhead of S2PRL]{
        \includegraphics[width=0.45\linewidth]{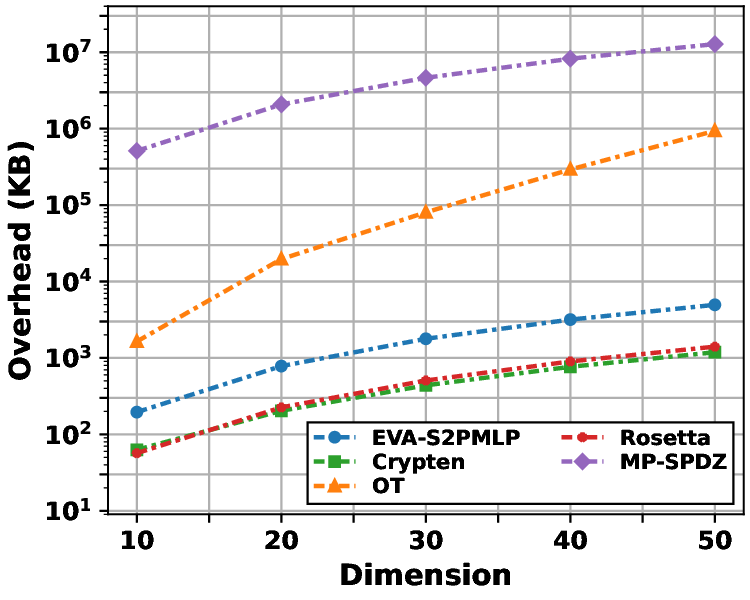}
        \Description{Communication overhead of S2PRL}
        \label{fig: S2PRL Communication}
    }
    \hspace{0.3cm}
    \subfigure[Communication overhead of S2PSM]{
        \includegraphics[width=0.45\linewidth]{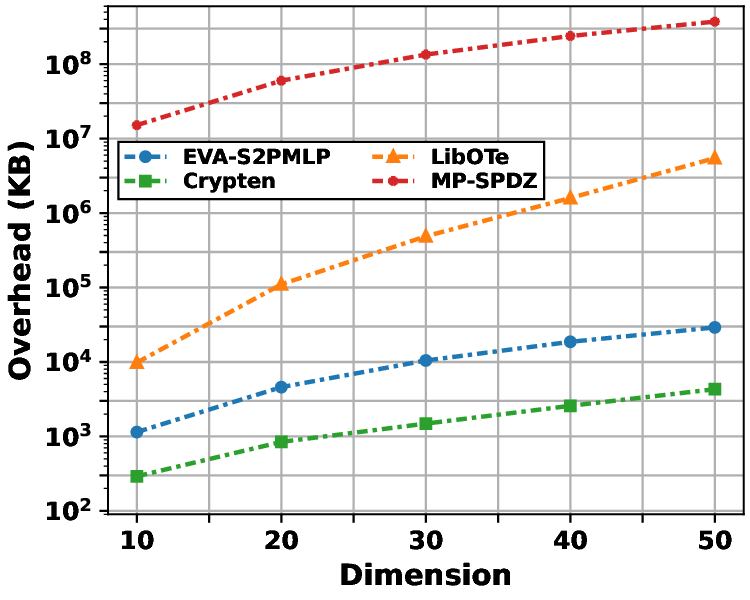}
        \Description{Communication overhead of S2PSM}
        \label{fig: S2PSM Communication}
    }
    \caption{Communication overhead comparison between EVA-S2PMLP and other frameworks}
    \label{fig: Communication Comparison}
    \vspace{-0.2cm}
\end{figure}

As illustrated in Figure~\ref{fig: Communication Comparison}, the communication volume of our framework ranks moderately high across the four basic protocols when compared with mainstream frameworks. This increase is mainly due to the overhead introduced by our integrated verification mechanisms.

\begin{figure}[htbp]
    \centering
    \subfigure[S2PHP execution time (LAN)]{
        \includegraphics[width=0.45\linewidth]{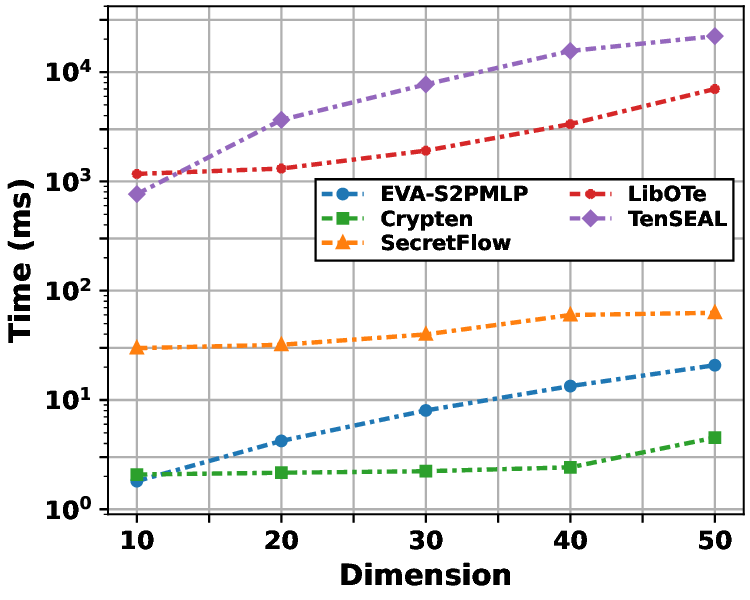}
        \Description{S2PHP execution time (LAN) plot}
        \label{fig: S2PHP Time LAN}
    }
    \hspace{0.3cm}
    \subfigure[S2PSCR execution time (LAN)]{
        \includegraphics[width=0.45\linewidth]{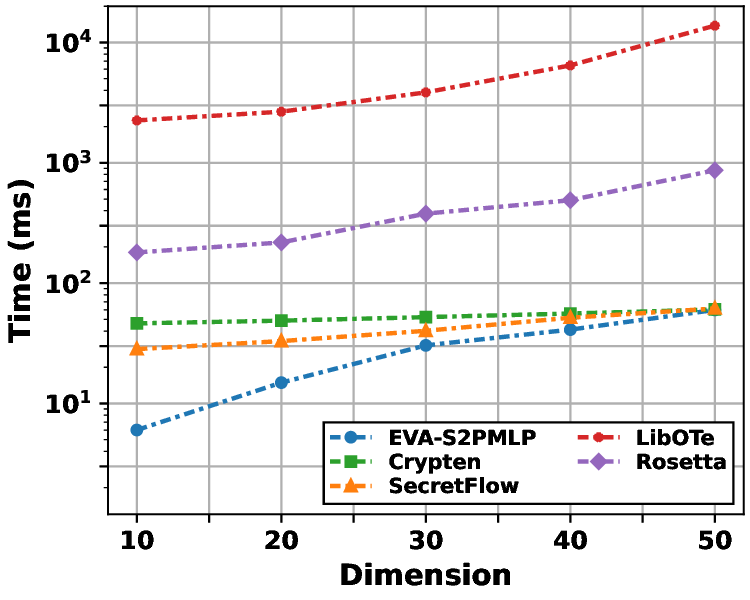}
        \Description{S2PSCR execution time (LAN) plot}
        \label{fig: S2PSCR Time LAN}
    }
    \subfigure[S2PRL execution time (LAN)]{
        \includegraphics[width=0.45\linewidth]{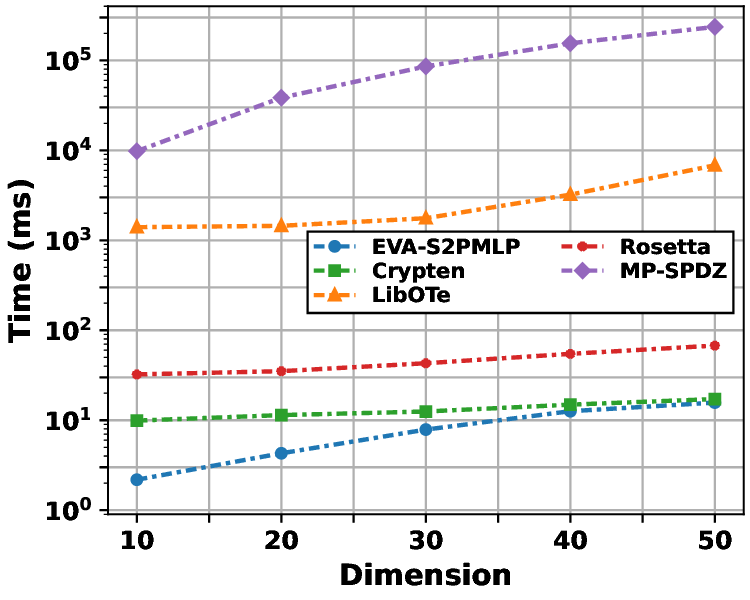}
        \Description{S2PRL execution time (LAN) plot}
        \label{fig: S2PRL Time LAN}
    }
    \hspace{0.3cm}
    \subfigure[S2PSM execution time (LAN)]{
        \includegraphics[width=0.45\linewidth]{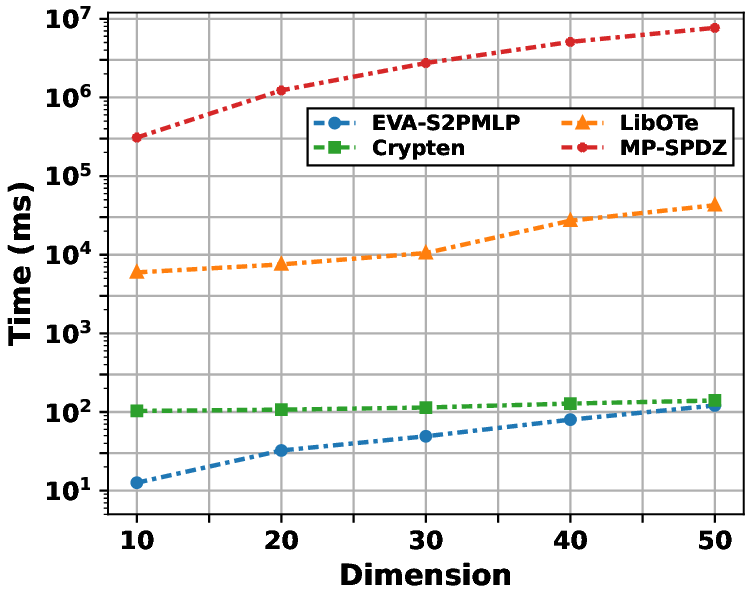}
        \Description{S2PSM execution time (LAN) plot}
        \label{fig: S2PSM Time LAN}
    }
    \caption{Total execution time comparison (LAN) between EVA-S2PMLP and other frameworks}
    \label{fig: LAN Time Comparison}
    \vspace{-0.2cm}
\end{figure}

Figure~\ref{fig: LAN Time Comparison} shows that under LAN settings, our framework achieves the fastest total execution time across all four basic protocols, with the exception of S2PHP, which ranks second. Crypten, SecretFlow, and Rosetta perform relatively well due to their use of secret sharing schemes, which incur minimal overhead during simple computations. LibOTe exhibits poor performance due to its reliance on encryption operations and high communication costs, and TenSEAL performs even worse as it operates entirely on encrypted data. Our MASCOT implementation based on MP-SPDZ, although providing high security through malicious OT, significantly compromises performance. Overall, our framework outperforms existing solutions due to fewer fixed communication rounds and streamlined computation.

\begin{figure}[htbp]
    \centering
    \subfigure[S2PHP execution time (WAN)]{
        \includegraphics[width=0.45\linewidth]{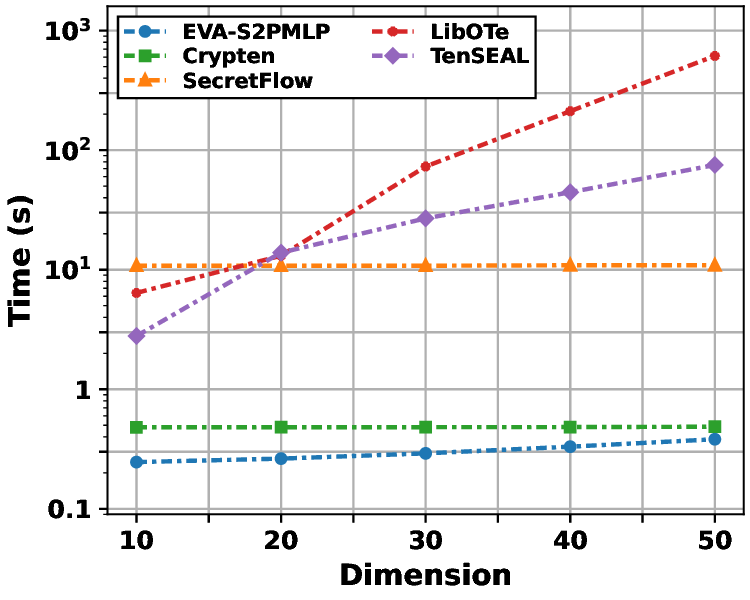}
        \Description{S2PHP execution time (WAN) figure}        
        \label{fig: S2PHP Time WAN}
    }
    \hspace{0.3cm}
    \subfigure[S2PSCR execution time (WAN)]{
        \includegraphics[width=0.45\linewidth]{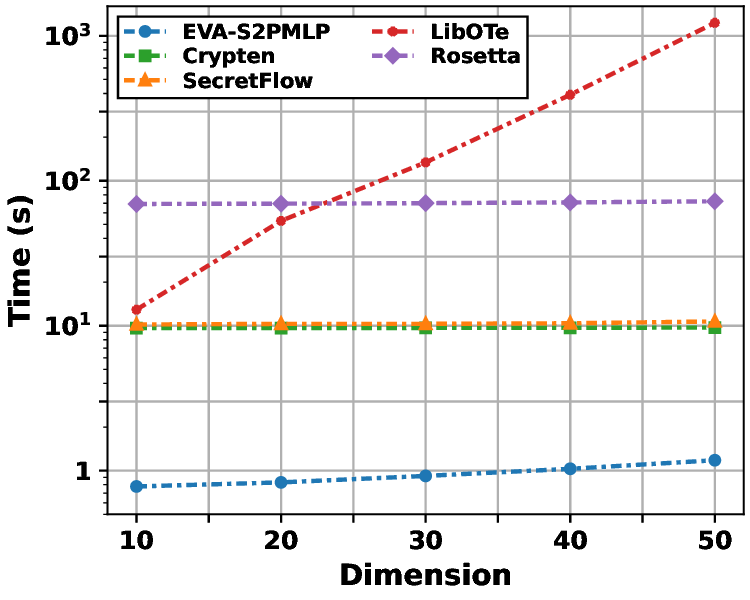}
        \Description{S2PSCR execution time (WAN) figure}        
        \label{fig: S2PSCR Time WAN}
    }
    \subfigure[S2PRL execution time (WAN)]{
        \includegraphics[width=0.45\linewidth]{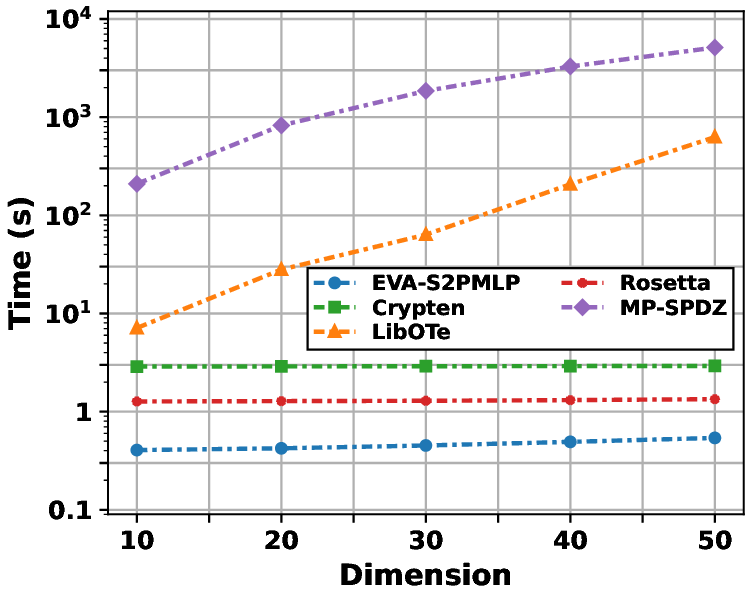}
        \Description{S2PRL execution time (WAN) figure}
        \label{fig: S2PRL Time WAN}
    }
    \hspace{0.3cm}
    \subfigure[S2PSM execution time (WAN)]{
        \includegraphics[width=0.45\linewidth]{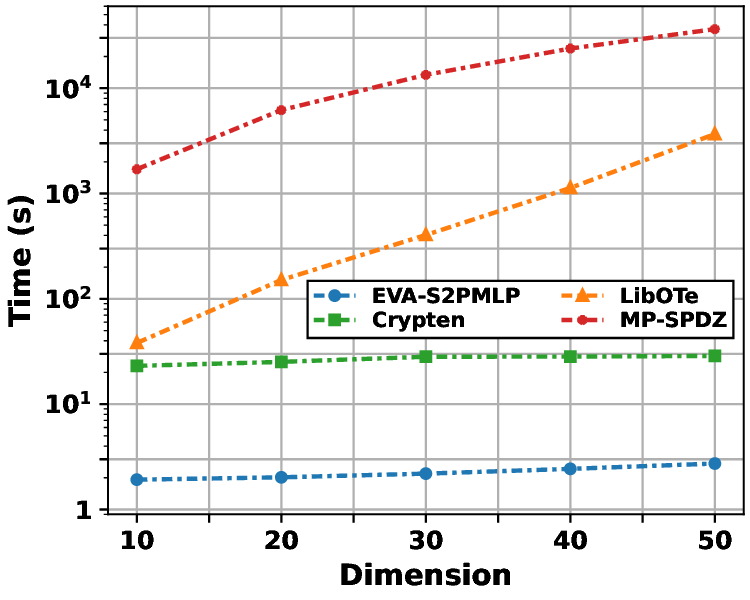}
        \Description{S2PSM execution time (WAN) figure}        
        \label{fig: S2PSM Time WAN}
    }
    \caption{Total execution time comparison (WAN) between EVA-S2PMLP and other frameworks}
    \label{fig: WAN Time Comparison}
    \vspace{-0.2cm}
\end{figure}

As shown in Figure~\ref{fig: WAN Time Comparison}, our framework consistently achieves the lowest execution time across all four protocols under WAN settings. This is largely due to our reduced communication rounds. In contrast, secret-sharing-based frameworks like Crypten require extensive message exchanges for share reconstruction, which leads to high latency in WAN environments.

\begin{figure}[htbp]
    \centering
    \subfigure[Verification time ratio of S2PHP]{
        \includegraphics[width=0.45\linewidth]{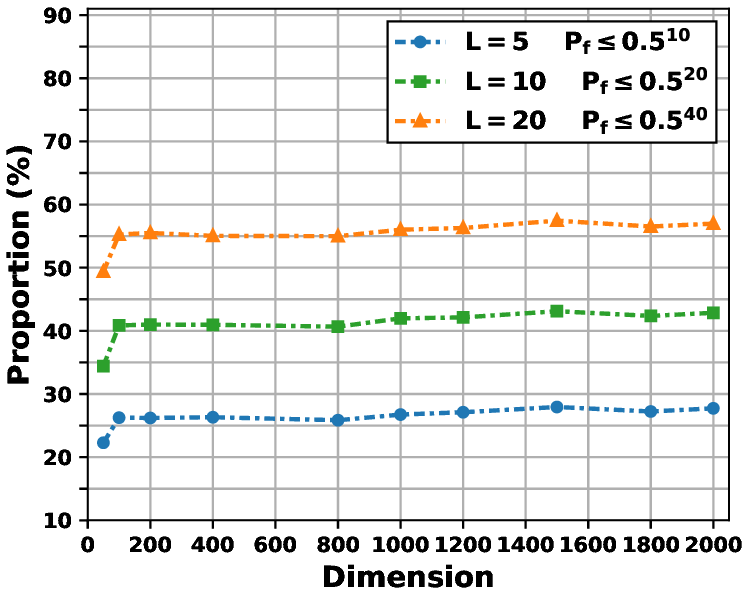}
        \Description{Verification time ratio plot for S2PHP protocol.}
        \label{fig: S2PHP Verification}
    }
    \hspace{0.3cm}
    \subfigure[Verification time ratio of S2PSCR]{
        \includegraphics[width=0.45\linewidth]{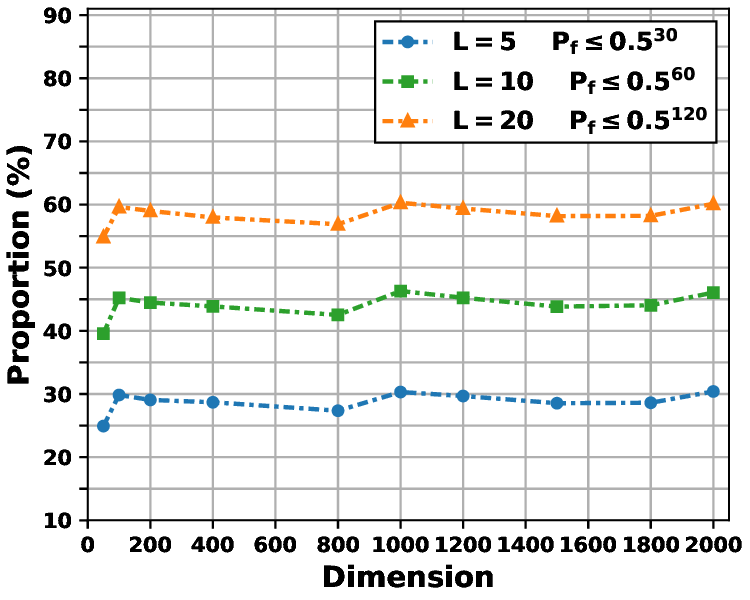}
        \Description{Verification time ratio plot for S2PSCR protocol.}        
        \label{fig: S2PSCR Verification}
    }
    \subfigure[Verification time ratio of S2PRL]{
        \includegraphics[width=0.45\linewidth]{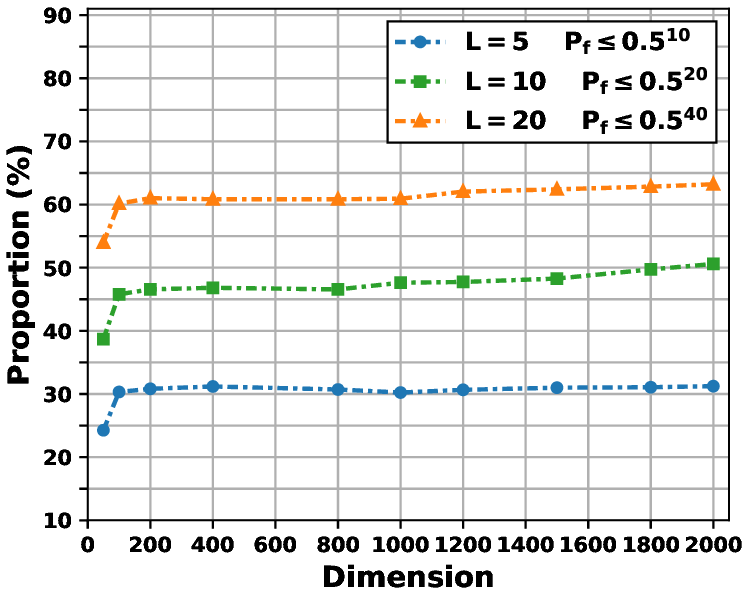}
        \Description{Verification time ratio plot for S2PRL protocol.}        
          \label{fig: S2PRL Verification}
    }
    \hspace{0.3cm}
    \subfigure[Verification time ratio of S2PSM]{
        \includegraphics[width=0.45\linewidth]{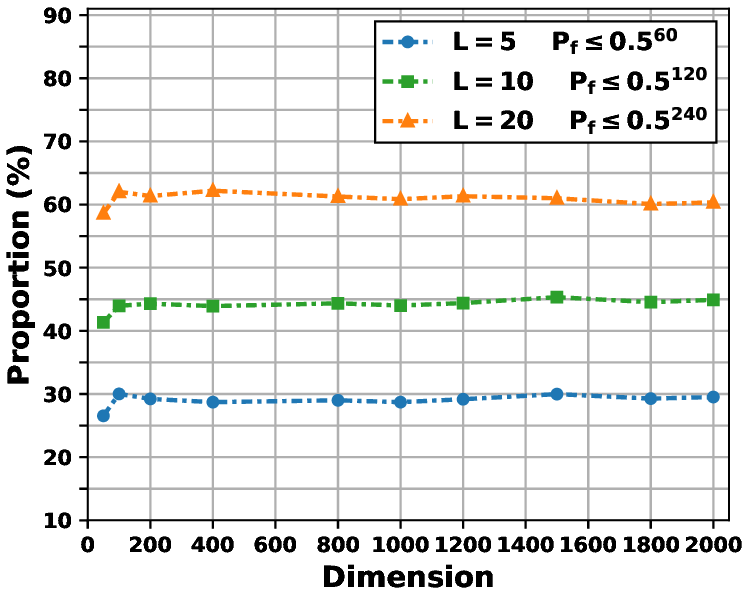}
        \Description{Verification time ratio plot for S2PSM protocol.}        
          \label{fig: S2PSM Verification}
    }
    \caption{Verification time ratio of basic protocols in EVA-S2PMLP}
    \label{fig: Verification Rate}
    \vspace{-0.2cm}
\end{figure}

From Figure~\ref{fig: Verification Rate}, we observe that the verification time ratio remains stable for large input dimensions but fluctuates slightly for small inputs. This is due to our batch optimization strategy, where both verification and total runtime scale linearly with the number of batches. The overall verification time ratio remains approximately equal to that of a single batch. When the verification round is set to $L=10$, the failure probabilities of the four protocols are negligible: $P_f(S2PHP)\approx9.53\times 10^{-7}$, $P_f(S2PSCR)\approx 8.67\times 10^{-19}$, $P_f(S2PRL)\approx 9.53\times 10^{-7}$, and $P_f(S2PSM)\approx 7.52\times 10^{-37}$. Thus, a verification round of $L=10$ is practically sufficient, and the verification overhead remains around $40\%$.

\begin{figure}[htbp]
    \centering
    \subfigure[Execution time for large matrices (LAN)]{
        \includegraphics[width=0.45\linewidth]{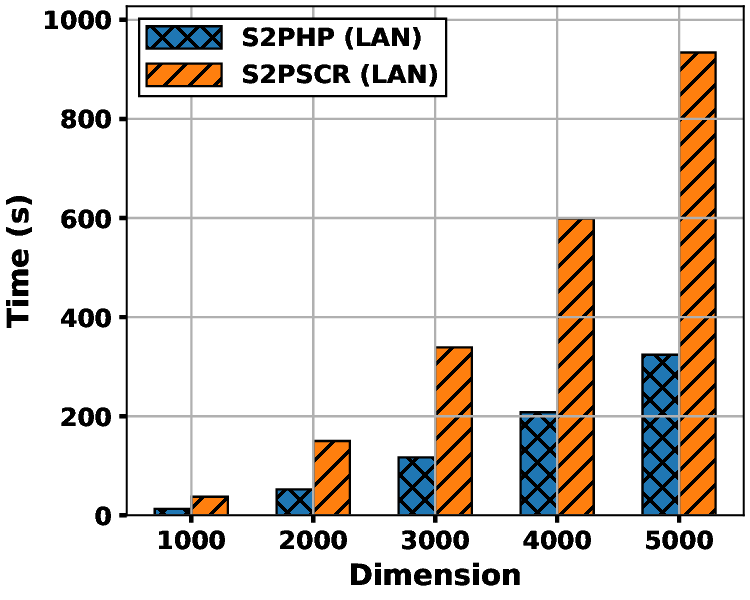}
        \label{fig: Huge_Matrix_LAN}
    }
    \hspace{0.3cm}
    \subfigure[Execution time for large matrices (WAN)]{
        \includegraphics[width=0.45\linewidth]{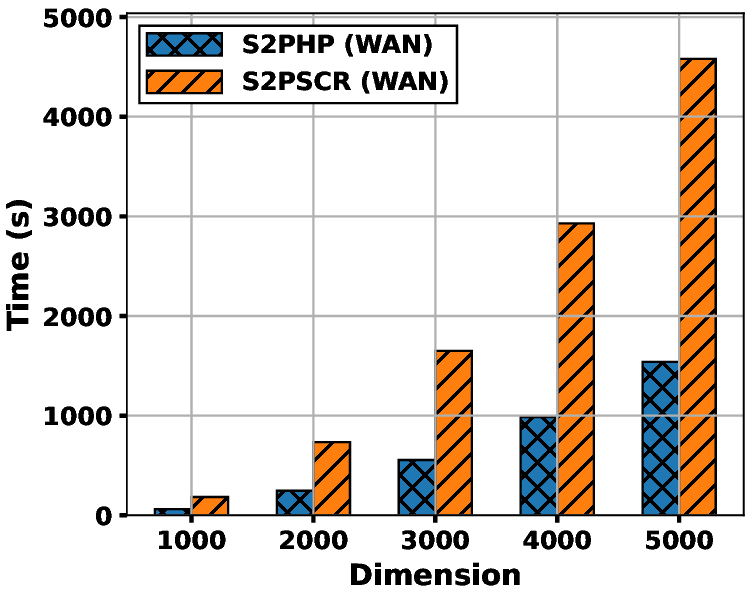}
        \label{fig: Huge_Matrix_WAN}
    }
    \caption{Performance of S2PHP and S2PSCR with large input matrices}
    \label{fig: Huge Matrix}
    \vspace{-0.2cm}
    \Description{Performance of S2PHP and S2PSCR with large input matrices.}
\end{figure}

Figure~\ref{fig: Huge Matrix} demonstrates that both S2PHP and S2PSCR maintain fast execution even for extremely large matrices. Under LAN settings with $5000 \times 5000$ input size, the total runtime for S2PHP is under 6 minutes, and under 16 minutes for S2PSCR. Moreover, their total execution time scales approximately linearly with input size, confirming the scalability and practicality of both protocols in real-world large-scale scenarios.

\subsection{Precision Evaluation of Basic Protocols}

In the precision evaluation experiments, we used $50\times50$ matrices as input and tested the maximum relative error of three core protocols in our framework (S2PHP, S2PSCR, and S2PSM) under five precision ranges $\delta$. We also compared the maximum relative errors of these protocols with those of other frameworks, as shown in Figure~\ref{fig: MRE Comparison}.

\begin{figure}[htbp]
    \centering
    \subfigure[Maximum relative error of the three protocols]
    {\includegraphics[width=0.45\linewidth]{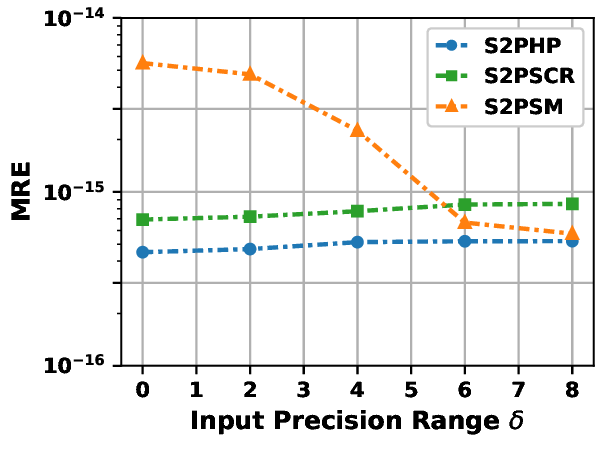}
    \label{fig: MRE of our Framework}}
    \hspace{0.3cm}
    \subfigure[Comparison of maximum relative error for S2PHP]{\includegraphics[width=0.45\linewidth]{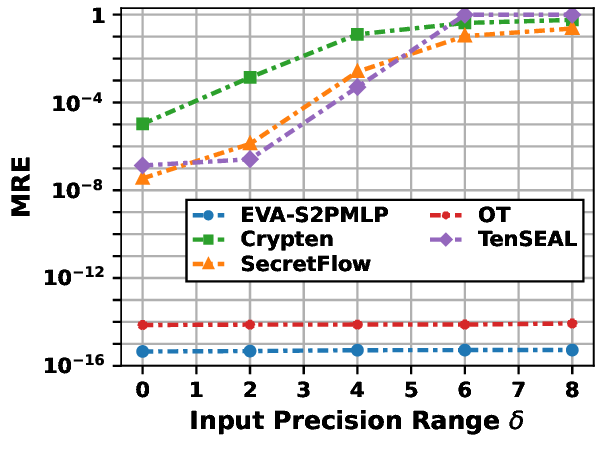}
    \label{fig: MRE Comparison of S2PHP}}

    \subfigure[Comparison of maximum relative error for S2PSCR]{\includegraphics[width=0.45\linewidth]{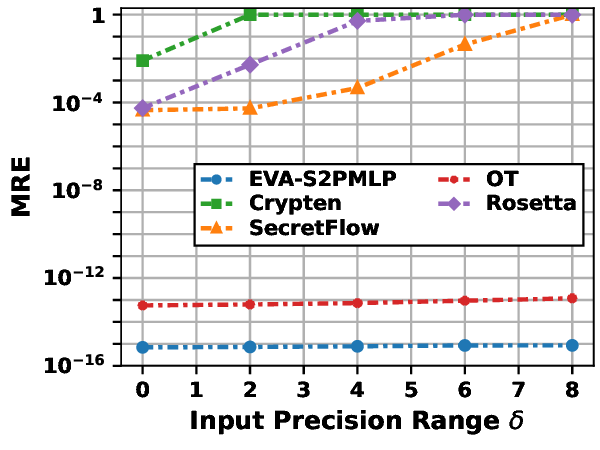}
    \label{fig: MRE Comparison of S2PSCR}}
    \hspace{0.3cm}
    \subfigure[Comparison of maximum relative error for S2PSM]{\includegraphics[width=0.45\linewidth]{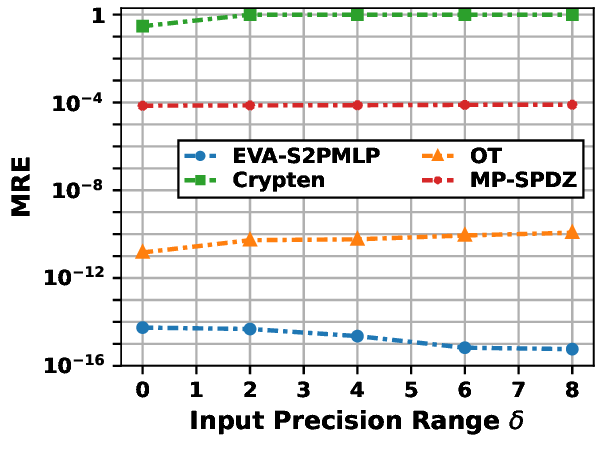}
    \label{fig: MRE Comparison of S2PSM}}

    \subfigure[Maximum relative error of S2PG]{\includegraphics[width=0.45\linewidth]{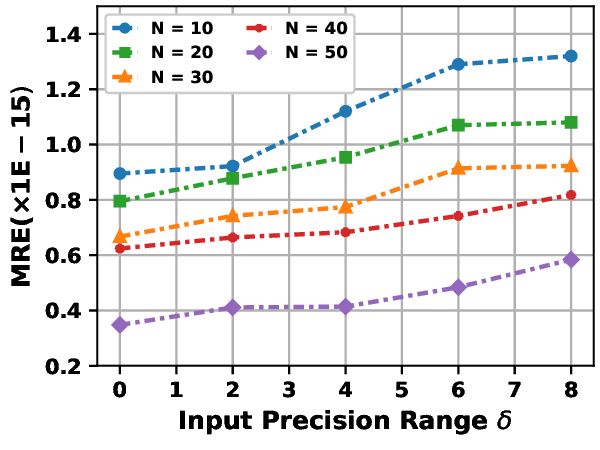}
    \label{fig: MRE of S2PG}}
    \caption{Precision comparison of EVA-S2PMLP basic protocols with other frameworks}
    \Description{A plot comparing the precision of EVA-S2PMLP basic protocols with other frameworks.}
    \label{fig: MRE Comparison}
    \vspace{-0.2cm}
\end{figure}

In our framework, the S2PHP and S2PSCR protocols can be viewed as parallel numerical operations over matrix elements. Hence, the precision can be analyzed from the perspective of numerical computation. As shown in Figure~\ref{fig: MRE Comparison}(a), the maximum relative error of S2PHP and S2PR increases slightly with the increase in $\delta$. This is because larger $\delta$ values result in greater disparities among input values, which in turn introduce more rounding errors during computation. However, this increase has an upper bound. Since S2PSCR is built upon S2PHP, and the core computation in S2PHP involves vector dot products, the upper bound of the relative error is given by the formula $MRE \le 1.25nu \frac{|\boldsymbol{x}|^T \cdot |\boldsymbol{y}|}{|{\boldsymbol{x}}^T\cdot \boldsymbol{y}|}$, as established in~\cite{Zarowski_2004}. Here, $n$ is the vector length (with split number $\rho=2$, $n=\rho^2=4$), $u$ is the unit roundoff (for 64-bit floating-point numbers, $u = 2^{-52}$), and the ratio $\frac{|\boldsymbol{x}|^T \cdot |\boldsymbol{y}|}{|{\boldsymbol{x}}^T\cdot \boldsymbol{y}|}$ equals 1 when random splits maintain the same sign. Hence, we derive the upper bound as $MRE \le 1.25 \times 4 \times 2^{-52} \approx 1.11 \times 10^{-15}$.

Interestingly, the maximum relative error of S2PSM decreases as $\delta$ increases. This is because the Softmax activation function is a normalized exponential function. When the data range varies widely, the exponentiation causes large values to dominate with probabilities close to 1, while small values shrink towards 0. As a result, larger $\delta$ values produce output matrices dominated by 0s and 1s, leading to a decrease in relative error.
As shown in Figures~\ref{fig: MRE Comparison}(b)$\sim$(d), our framework consistently outperforms others in terms of accuracy across S2PHP, S2PSCR, and S2PSM. LibOTe ranks second, since it mainly secures data during transmission and only incurs minor precision loss during data splitting. In the case of S2PHP, TenSEAL operates over the ring $\mathbb{Z}_{2^k}$, and achieves stable accuracy under small $\delta$ ranges. However, once the values exceed the representation range, the error grows rapidly. Garbled circuits introduce large errors during conversions between arithmetic and Boolean circuits. Similarly, secret-sharing-based schemes suffer substantial precision loss due to their use of fixed-point representations. Consequently, frameworks such as Crypten, SecretFlow, and Rosetta exhibit poor performance.
Specifically for S2PSM, MP-SPDZ achieves relatively stable accuracy by transforming exponential functions into partial integer and floating-point operations via bit decomposition, and then approximating floating-point computations using polynomials. In contrast, other frameworks generally use polynomial or piecewise approximations for activation functions, leading to significant loss in accuracy. In conclusion, the precision comparison experiments demonstrate that our framework offers a substantial accuracy advantage over existing mainstream frameworks.

\subsection{Performance Evaluation of Secure Two-Party MLP}

To evaluate the performance of the EVA-S2PMLP framework, we conducted training and inference tasks using vertically partitioned secure two-party multilayer perceptron (MLP) models on different datasets and compared them with two state-of-the-art frameworks (Crypten and FATE) in terms of time and various evaluation metrics. Additionally, we implemented a plaintext MLP using Python (PlainMLP) as a performance baseline, and used it to standardize the parameter configuration for all models.

\textbf{Datasets and Experimental Setup:}  
We used two small datasets, \textbf{Iris} (150 samples: 120 for training, 30 for testing, with 4 features per sample) and \textbf{Wine} (178 samples: 142 for training, 36 for testing, with 13 features per sample), as well as two large datasets, \textbf{MNIST} (70,000 samples: 60,000 for training, 10,000 for testing, with 784 features per sample) and \textbf{CIFAR} (60,000 samples: 50,000 for training, 10,000 for testing, with 3,072 features per sample). The two small datasets are for 3-class classification, and the two large datasets are for 10-class classification. In secure two-party MLP training and inference, each dataset’s features are equally divided between Alice and Bob as private input, while the labels are considered public. We adopted a three-layer MLP model (one hidden layer and one output layer) for all datasets. For the small datasets, the hidden layer size and batch size are both 16, and the learning rate is 0.1. For the large datasets, both values are set to 128, with a learning rate of 0.01. All models are trained for 5 epochs.

\textbf{Evaluation Metric:}  
In multi-class classification using MLPs, suppose the number of classes is $d$, then the label of each sample is represented as a one-hot vector of length $d$, where the index of the value '1' indicates the correct class. The MLP outputs a $d$-dimensional float vector for each test sample, representing the probability of belonging to each class. The class with the maximum probability is taken as the predicted label. Model performance is evaluated using accuracy, defined as the proportion of correctly predicted labels in the test set.

We compared our framework with FATE and Crypten across all four datasets, using the plaintext model PlainMLP as a reference. The time comparison results are shown in Table~\ref{tab:S2PMLP-Time}, and accuracy results are shown in Table~\ref{tab:S2PMLP-Accuracy}. Moreover, we analyzed the stage-wise time breakdown of training and inference in our framework on these datasets, as shown in Figure~\ref{fig: EVA-S2PMLP Stage Time}.

\begin{table}[htbp]
  \centering
  \caption{Efficiency Comparison of Secure Two-Party MLP Models}
  \vspace{1mm}
  \resizebox{\linewidth}{!}{
    \begin{tabular}{cccccccccccccccc}
    \toprule  
    \multicolumn{1}{c}{\multirow{2}[2]{*}{\textbf{Dataset}}} & 
    \multicolumn{1}{c}{\multirow{2}[2]{*}{\textbf{Framework}}} & 
    \multicolumn{2}{c}{\textbf{Training Overhead}} & & 
    \multicolumn{2}{c}{\textbf{Training Time (s)}} & & 
    \multicolumn{2}{c}{\textbf{Inference Overhead}} & & 
    \multicolumn{2}{c}{\textbf{Inference Time (s)}} \\
    \cmidrule{3-4} \cmidrule{6-7} \cmidrule{9-10} \cmidrule{12-13}
    &  & \textbf{Comm. (MB)} & \textbf{Rounds} & & \textbf{LAN} & 
    \textbf{WAN} & & \textbf{Comm. (MB)} & \textbf{Rounds} & & \textbf{LAN} & \textbf{WAN} \\
    \midrule
    \multirow{3.5}[2]{*}{\textbf{Iris}}
    & \textbf{Crypten} & 34.32 & 53336 & & 7.57  & 2136.57  & & 0.78 & 212 & & 0.09  & 8.57  \\
    & \textbf{Fate} & 7.35 & 6280 & & 99.86  & 350.62 & & 0.17 & 685 & & 23.24 & 50.58 \\
    & \textbf{EVA-S2PMLP} & 156.04 & 9960 & & 1.68 & 403.13 & & 3.23 & 117 & & 0.02 & 4.78 \\
    & \textbf{PlainMLP} & - & - & & \textbf{2.95E-03} & - & & - & - & & \textbf{3.65E-05} & - \\
    \midrule
    \multirow{3.5}[2]{*}{\textbf{Wine}}
    & \textbf{Crypten} & 42.79  & 59916  & & 9.60 & 2401.36  & & 0.97  & 279  & & 0.1 & 11.26  \\
    & \textbf{Fate} & 11.20  & 9429  & & 115.99  & 492.50  & & 0.25  & 772  & & 27.30  & 58.11  \\
    & \textbf{EVA-S2PMLP} & 223.55  & 11205  & & 2.04 & 454.90 & & 4.53  & 117 & & 0.03 & 4.81 \\
    & \textbf{PlainMLP} & - & - & & \textbf{3.23E-03} & - & & - & - & & \textbf{4.12E-05} & - \\
    \midrule
    \multirow{3.5}[2]{*}{\textbf{Mnist}}
    & \textbf{Crypten} & 156982.25  & 5034723  & & 7350.09  & 212300.31  & & 3552.21  & 921 & & 71.01 & 699.73  \\
    & \textbf{Fate} & 5575.10 & 4498933  & & 158569.63  & 338221.41  & & 126.15 & 1372 & & 544.33 & 602.34 \\
    & \textbf{EVA-S2PMLP} & 2171802.23  & 583905  & & 6439.94 & 85972.55 & & 22664.37  & 117 & & 65.49  & 657.01 \\
    & \textbf{PlainMLP} & - & - & & \textbf{20.75} & - & & - & - & & \textbf{0.1} & - \\
    \midrule
    \multirow{3.5}[2]{*}{\textbf{Cifar}}
    & \textbf{Crypten} & 258809.20 
 & 7907699  & & 23120.84  & 345339.41  & & 5856.36  & 1082  & & 3039.30 & 5234.11  \\
    & \textbf{Fate} & 10798.50 & 7570978 & & 178975.13 & 481336.76 & & 244.35 & 1694 & & 815.21 & 889.13 \\
    & \textbf{EVA-S2PMLP} & 7738807.30  & 486795  & & 19595.81 
 & 239401.05 & & 70402.87  & 117 & & 2975.50  & 4803.12 \\
    & \textbf{PlainMLP} & - & - & & \textbf{117.04} & - & & - & - & & \textbf{1.09} & - \\
    \bottomrule
    \end{tabular}}%
  \label{tab:S2PMLP-Time}%
\end{table}%

From Table~\ref{tab:S2PMLP-Time} and Table~\ref{tab:S2PMLP-Accuracy}, we observe that our framework delivers the best overall performance in terms of both time efficiency and model accuracy. FATE, which is based on federated learning, has relatively low communication volume, but incurs large computational overhead due to the use of homomorphic encryption. Crypten, based on secret sharing, leads to a very large number of communication rounds, which impacts overall runtime. Although our framework has a higher communication volume, the number of communication rounds is much smaller and the computation is lightweight, resulting in low time cost. Moreover, the inability of FATE and Crypten to accurately compute non-linear operations leads to inferior accuracy compared to our framework. Overall, EVA-S2PMLP outperforms existing frameworks in secure MLP tasks.

\begin{table}[htbp]
  \centering
  \caption{Accuracy Comparison of Secure Two-Party MLP Models}
  \vspace{1mm}
  \resizebox{!}{!}{
    \begin{tabular}{cccccccccccc}
    \toprule  
    \multicolumn{1}{c}{\multirow{2}[2]{*}{\textbf{Framework}}} & & 
    \multicolumn{7}{c}{\textbf{Accuracy}} \\
    \cmidrule{3-9} &  & \textbf{Iris} & & \textbf{Wine} & & \textbf{MNIST} & & 
    \textbf{CIFAR}  \\
    \midrule
    \textbf{Crypten} & & 0.9333 & & 0.9722 & & 0.9325 & & 0.4627  \\
    \textbf{Fate} & & 0.8000 & & 0.9444 & & 0.9332 & & 0.3729  \\
    \textbf{EVA-S2PMLP} & & \textbf{1.0000} & & \textbf{1.0000} & & \textbf{0.9582} & & \textbf{0.4812}  \\
    \textbf{PlainMLP} & & \textbf{1.0000} & & \textbf{1.0000} & & \textbf{0.9638} & & \textbf{0.4951}  \\
    \bottomrule
    \end{tabular}}%
  \label{tab:S2PMLP-Accuracy}%
\end{table}%

As shown in Figure~\ref{fig: EVA-S2PMLP Stage Time}, under a LAN environment, the main components of model training and prediction time are offline computation, online computation, and communication. For large datasets (MNIST and CIFAR), the computation time exceeds communication time due to the high bandwidth and low latency of LANs. Notably, the verification time accounts for a very small proportion in MLP applications, demonstrating the lightweight nature of our verification module.

\begin{figure}[htbp]
    \centering

    \subfigure[S2PMLP-TR Time Breakdown]{\includegraphics[width=0.45\linewidth]{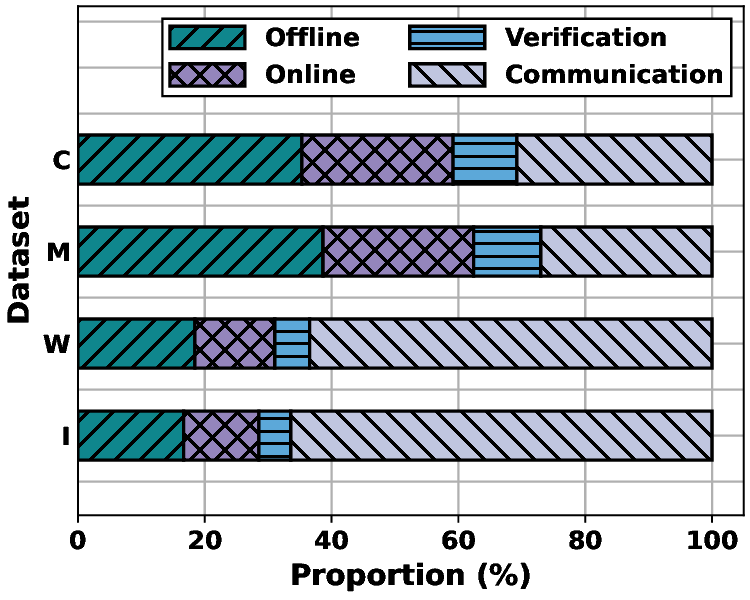}
    \Description{Bar chart showing the time breakdown for S2PMLP-TR in different stages.}
    \label{fig: S2PMLP-TR Stage Time}}
    \hspace{0.3cm}
    \subfigure[S2PMLP-PR Time Breakdown]{\includegraphics[width=0.45\linewidth]{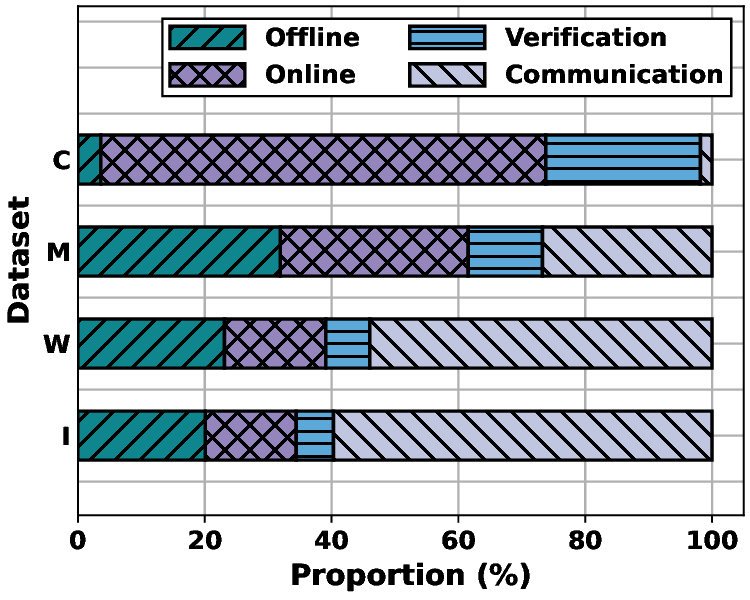}
    \Description{Bar chart showing the time breakdown for S2PMLP-PR in different stages.}
    \label{fig: S2PMLP-PR Stage Time}}
    \caption{Time Breakdown of S2PMLP-TR and S2PMLP-PR in Different Stages}
    \label{fig: EVA-S2PMLP Stage Time}
    \vspace{-0.2cm}
\end{figure}

\section{Conclusion and future work on EVA-S2PMLP}\label{Conclusion-S2PMLP}

\noindent In this chapter, we proposed \texttt{EVA-S2PMLP}, a privacy-preserving framework for secure two-party multilayer perceptron (MLP) training and inference, which achieves superior overall performance across multiple datasets compared to state-of-the-art privacy-preserving frameworks such as FATE, CrypTen, and SecretFlow. We introduced several atomic protocols—S2PHP, S2PDRL, S2PRL, and S2PSM—for secure matrix operations and nonlinear activation functions, and conducted a thorough experimental evaluation of their communication cost, computational efficiency, and numerical precision.
The S2PRL and S2PSM protocols in \texttt{EVA-S2PMLP} adopt exact implementations of nonlinear functions, such as ReLU and Softmax, providing significant precision advantages over approximation-based counterparts used in most mainstream frameworks. Our S2PHP (matrix multiplication) and S2PDRL (dropout-based regularization) protocols introduce a configurable splitting parameter $\rho$ to balance computation and communication cost; in this work, we set $\rho=2$ to maximize efficiency while maintaining precision. Furthermore, we extended the underlying vector verification mechanism to support matrix-level verification, ensuring correctness in the presence of floating-point computations. However, this extension introduces additional computational and communication overhead, highlighting the need for more lightweight verification techniques in future work.
Overall, \texttt{EVA-S2PMLP} achieves near-baseline accuracy in privacy-preserving MLP tasks while significantly reducing communication rounds and improving runtime efficiency during both training and inference. Experimental results on datasets ranging from small-scale (Iris, Wine) to large-scale (MNIST, CIFAR) confirm the robustness, scalability, and practical effectiveness of our approach. These results suggest that \texttt{EVA-S2PMLP} holds strong potential for real-world privacy-preserving machine learning applications in distributed environments such as federated healthcare and financial analytics.

\noindent In future work, we plan to: (1) design more efficient matrix-level verification protocols to reduce overhead in high-dimensional MLPs; (2) extend the framework to support active adversaries and malicious threat models, further enhancing its security guarantees; and (3) evaluate its deployment under realistic WAN settings involving multiple independent institutions. Additionally, generalizing the current 2-party protocol design to arbitrary $n$-party settings remains an important direction toward achieving broader applicability and system scalability.

\section{CRediT authorship contribution statement}
\noindent Shizhao Peng: Conceptualization, Methodology, Software, Validation, Formal analysis, Visualization, Writing-original draft, Writing-review \& editing. 
Shoumo Li: Software, Validation.
Tianle Tao: Software, Writing-review \& editing. 
\section{Acknowledgments}
This work was supported by the National Science and Technology Major Project of the Ministry of Science and Technology of China Grant No.2022XAGG0148.

\bibliographystyle{unsrt}

\bibliography{reference}

\appendix

\section{Algorithms of EVA-S3PC}
\label{APP: Algorithms In EVA-S3PC}
This section introduces the S2PM and S2PHM protocols in EVA-S3PC\cite{EVA-S3PC}.
\subsection{S2PM}
The problem definition of S2PM is as follows: 

\begin{problem}[Secure 2-Party Matrix Multiplication]\label{Problem-S2PM}
    Alice has an $n\times s$ matrix $A$ and Bob has an $s\times m$ matrix $B$. They want to conduct the multiplication, such that Alice gets $V_a$ and Bob gets $V_b$, where $V_a+V_b=A\times B$.
\end{problem}

\subsubsection{Description of S2PM}
The S2PM includes three stages: CS pre-processing stage  in Algorithm \ref{alg:S2PM-Preprocessing}, online computation stage in Algorithm \ref{alg:S2PM-Computing}, and result verification stage in Algorithm \ref{alg:S2PM-Verification}.\\

\textbf{Pre-processing Stage.}
In Algorithm \ref{alg:S2PM-Preprocessing}, CS generates random private matrices $R_a$ for Alice and $R_b$ for Bob, where $rank(R_a) < s$ and $rank(R_b) < s$. Subsequently, CS computes $S_t=R_a \cdot R_b$ and generates random matrices $r_a$ for Alice and $r_b$ for Bob, where $r_a + r_b = S_t$. Finally, CS sends a set of matrices $(R_a, r_a, S_t)$ to Alice and $(R_b, r_b, S_t)$ to Bob.

\begin{breakablealgorithm}
    \caption{S2PM CS Pre-processing Stage}
    \label{alg:S2PM-Preprocessing}
    \begin{algorithmic}[1] 
        \Require $n, s, m$
        \Ensure Alice $\Leftarrow$ $(R_a,r_a,S_t)$ and Bob $\Leftarrow$ $(R_b,r_b,S_t)$
        \State $R_a \gets$ generate a random matrix \Comment{$R_a \in \mathbb{R}^{n \times s}, rank(R_a)=min(n,s)-1$}
        \State $R_b \gets$ generate a random matrix \Comment{$R_b \in \mathbb{R}^{s \times m}, rank(R_b)=min(s,m)-1$}
        \State $S_t = R_a \times R_b$ \Comment{$S_t \in \mathbb{R}^{n \times m}$}
        \State $r_a,r_b \gets$ generate random matrices \Comment{$r_a,r_b \in \mathbb{R}^{n \times m}, r_a + r_b = S_t$}
        \State Alice $\Leftarrow(R_a, r_a, S_t)$
        \State Bob $\Leftarrow(R_b, r_b, S_t)$
        \State \Return $(R_a, r_a, S_t)$, $(R_b, r_b, S_t)$
    \end{algorithmic}
\end{breakablealgorithm}

\textbf{Online Stage.}
In Algorithm \ref{alg:S2PM-Computing}, Alice computes $\hat{A} = A + R_a$ and sends $\hat{A}$ to Bob while Bob computes $\hat{B} = B + R_b$ and sends $\hat{B}$ to Alice. Bob then generates a random matrix $V_b$, computes $VF_b=V_b-\hat{A}\times B$, $T = r_b - VF_b$, and then sends  $(VF_b,T)$ to Alice. Finally, Alice computes the matrix $V_a = T + r_a - (R_a \times \hat{B})$, $VF_a = V_a + R_a \times \hat{B}$, and sends $VF_a$ to Bob.

\begin{breakablealgorithm}
    \caption{S2PM Online Computing Stage}
    \label{alg:S2PM-Computing}
    \begin{algorithmic}[1] 
        \Require $A \in \mathbb{R}^ {n \times s}$ and $B \in \mathbb{R}^{s\times m}$
        \Ensure Alice $\Leftarrow(V_a,VF_a)$ and Bob $\Leftarrow(V_b,VF_b)$
        \State $\hat{A} = A + R_a$ and send $\hat{A} \Rightarrow$ Bob \Comment{$\hat{A}\in \mathbb{R}^ {n \times s}$}
        \State $\hat{B} = B + R_b$ and send $\hat{B} \Rightarrow$ Alice \Comment{$\hat{B}\in \mathbb{R}^ {s \times m}$}
        \State $V_b\gets$ generate a random matrix \Comment{$V_b\in \mathbb{R}^ {n \times m}$}
        \State $VF_b = V_b-\hat{A}\times B$ \Comment{$VF_b\in \mathbb{R}^ {n \times m}$}
        \State $T = r_b - VF_b$ \Comment{$T\in \mathbb{R}^ {n \times m}$}
        \State send $(VF_b,T)\Rightarrow$ Alice
        \State $V_a = T + r_a - (R_a \times \hat{B})$ \Comment{$V_a\in \mathbb{R}^ {n \times m}$}
        \State $VF_a=V_a+R_a\times \hat{B}$ and send $VF_a\Rightarrow$ Bob \Comment{$VF_a\in \mathbb{R}^ {n \times m}$}
        \State \Return $(V_a,VF_a)$, $(V_b,VF_b)$ 
    \end{algorithmic}
\end{breakablealgorithm}

\textbf{Verification Stage.}
In Algorithm \ref{alg:S2PM-Verification}, Alice and Bob perform the same steps for $l$ rounds of verification. In each round, a vector $\hat{\delta_a}$ whose elements are all randomly composed of 0 or 1 is generated for the computation of $E_r=(VF_a+VF_b-S_t)\times \hat{\delta_a}$. Accept if $E_r = (0,0,\cdots,0)^T$ holds for all $l$ rounds, reject otherwise.

\begin{breakablealgorithm}
    \caption{S2PM Result Verification Stage}
    \label{alg:S2PM-Verification}
    \begin{algorithmic}[1]
        \Require $VF_a, VF_b, S_t \in \mathbb{R}^ {n \times m}$ and $l > 0$
        \Ensure Accept if verified, Reject otherwise
        \For{$i := 1$ to $l$}
            \State $\hat{\delta_a}\gets$ generate a vector randomly composed of 0 or 1 \Comment{$\hat{\delta_a} \in \mathbb{R}^{m\times 1}$} 
            \State  $E_r=(VF_a+VF_b-S_t)\times \hat{\delta_a}$ \Comment{$E_r \in \mathbb{R}^{n\times 1}$}
            \If{$E_r\neq (0,0,\cdots,0)^T$}
                \State \Return Rejected; 
            \EndIf
        \EndFor
        \State \Return Accepted
    \end{algorithmic}
\end{breakablealgorithm}

\subsection{S2PHM}\label{S2PHM}
The problem definition of S2PHM is as follows:
\begin{problem}[Secure 2-Party Matrix Hybrid Multiplication Problem]
    Alice has private matrices $(A_1,A_2)$ and Bob has private matrices $(B_1,B_2)$, where $(A_1,B_1) \in \mathbb{R}^{n\times s}$,$(A_2,B_2) \in \mathbb{R}^{s\times m}$. They want to conduct the hybrid multiplication $f[(A_1, A_2), (B_1, B_2)] = (A_1 + B_1) \cdot (A_2+B_2)$ in which Alice gets $V_a$ and Bob gets $V_b$ such that $V_a + V_b =  (A_1 + B_1) \cdot (A_2+B_2)$.
\end{problem}

\subsubsection{Description of S2PHM}
In S2PHM, Alice and Bob compute $V_{a0} = A_1 \times A_2, V_{b0} = B_1 \times B_2$ respectively, and then jointly compute $V_{a1} + V_{b1} = A_1 \times B_2$  and $V_{b2}$, $V_{a2} \in \mathbb{R}^{n \times m}$ with S2PM protocol. Finally, Alice sums $V_a = V_{a0} + V_{a1} + V_{a2}$ and Bob sums $V_b = V_{b0} + V_{b1} + V_{b2}$.

\begin{breakablealgorithm}
    \caption{S2PHM}%
    \label{alg:S2PHM}
    \begin{algorithmic}[1] 
        \Require $(A_1,B_1) \in \mathbb{R}^{n\times s}$, $(A_2,B_2) \in \mathbb{R}^{s\times m}$
        \Ensure $V_a, V_b \in\mathbb{R}^ {n \times m}$ 
        \State $V_{a0} = A_1 \times A_2$ \Comment{$V_{a0} \in \mathbb{R}^{n \times m}$}
        \State $V_{b0} = B_1 \times B_2$ \Comment{$V_{b0} \in \mathbb{R}^{n \times m}$}
        \State $V_{a1}, V_{b1}\gets$ S2PM$(A_1, B_2)$ \Comment{$V_{a1}, V_{b1} \in \mathbb{R}^{n \times m}$}
        \State $V_{b2}, V_{a2}\gets$ S2PM$(B_1, A_2)$ \Comment{$V_{a2}, V_{b2} \in \mathbb{R}^{n \times m}$}
        \State $V_a = V_{a0} + V_{a1} + V_{a2}$ \Comment{$V_a \in \mathbb{R}^{n \times m}$}
        \State $V_b = V_{b0} + V_{b1} + V_{b2}$ \Comment{$V_b \in \mathbb{R}^{n \times m}$}
        \State \Return $V_a,V_b$ 
    \end{algorithmic}
\end{breakablealgorithm}

\end{sloppypar}
\end{document}